%% file: main.tex
\documentclass[letterpaper, 10 pt, conference]{ieeeconf}
\IEEEoverridecommandlockouts
\usepackage{cite}
\usepackage{amsmath,amssymb,amsfonts}
\usepackage{algorithmic}
\usepackage{graphicx}
\usepackage{textcomp}
\usepackage{xcolor}
\usepackage{mathrsfs}
\usepackage{psfrag}
\usepackage{cancel}

\def\mc{\mathcal}

\begin{document}
\title{\LARGE \bf Negative Imaginary State Feedback Equivalence for Systems of Relative Degree One and Relative Degree Two
}

\author{Kanghong Shi,$\quad$Ian R. Petersen, \IEEEmembership{Fellow, IEEE},$\quad$and$\quad$Igor G. Vladimirov 
\thanks{This work was supported by the Australian Research Council under grant DP190102158.}
\thanks{K. Shi, I. R. Petersen and I. G. Vladimirov are with the School of Engineering, College of Engineering and Computer Science, Australian National University, Canberra, Acton, ACT 2601, Australia.
        {\tt kanghong.shi@anu.edu.au}, {\tt ian.petersen@anu.edu.au}, {\tt igor.vladimirov@anu.edu.au}.}%
}

\newtheorem{definition}{Definition}
\newtheorem{theorem}{Theorem}
\newtheorem{conjecture}{Conjecture}
\newtheorem{lemma}{Lemma}
\newtheorem{remark}{Remark}
\newtheorem{corollary}{Corollary}
\newtheorem{assumption}{Assumption}

\maketitle
\thispagestyle{plain}
\pagestyle{plain}

\begin{abstract}
\input{Abstract.tex}
\end{abstract}

\begin{keywords}
Negative imaginary systems, feedback equivalence, stabilization, robust control.
\end{keywords}

\input{Introduction.tex}
\input{Notation.tex}
\input{Preliminaries.tex}

\input{Feedback_equivalence_nonlinear_NI.tex}

\input{Conclusion.tex}

\bibliographystyle{IEEEtran}
\input{main.bbl}

\end{document}

%% file: Abstract.tex
This paper presents necessary and sufficient conditions under which a linear system of relative degree either one or two is state feedback equivalent to a negative imaginary (NI) system. More precisely, we show for a class of linear time-invariant strictly proper systems, that such a system can be rendered minimal and NI using full state feedback if and only if it is controllable and weakly minimum phase. A strongly strict negative imaginary state feedback equivalence result is also provided. The NI state feedback equivalence result is then applied in a robust stabilization problem for an uncertain system with a strictly negative imaginary uncertainty.

%% file: Introduction.tex
\section{INTRODUCTION}
The negative imaginary (NI) systems theory was established in \cite{lanzon2008stability,petersen2010feedback,xiong2010negative} to address robust control problems for systems with colocated force actuators and position sensors. Motivated by the control of flexible structures \cite{preumont2018vibration,halim2001spatial,pota2002resonant}, negative imaginary systems theory has been applied in many fields. For example, it has achieved success in the control of lightly damped structures \cite{cai2010stability,rahman2015design,bhikkaji2011negative} and nano-positioning \cite{mabrok2013spectral,das2014mimo,das2014resonant,das2015multivariable}. NI systems theory provides an important alternative to positive real (PR) systems theory \cite{brogliato2007dissipative} in order to achieve robust stability when the output of a mechanical system is position rather than velocity. Unlike PR systems theory which uses negative velocity feedback control, NI systems can be controlled using positive position feedback. One advantage of NI systems theory is that it can deal with systems of relative degrees zero, one  or two, while PR systems theory can only deal with systems of relative degrees zero or one \cite{brogliato2007dissipative}.

Roughly speaking, a square transfer matrix is NI if it is stable and its Hermitian imaginary part is negative semidefinite for all frequencies $\omega \geq 0$. In particular, the frequency response of a single-input single-output (SISO) NI system has a phase lag between $0$ to $-\pi$ radians for all frequencies $\omega>0$. The negative imaginary lemma states that a system is NI if it is dissipative, with its supply rate being the inner product of its input and the derivative of its output \cite{xiong2010negative,song2012negative,mabrok2020dissipativity}. Also note that for the positive feedback interconnection of an NI system $R(s)$ and a strictly negative imaginary (SNI) system $R_s(s)$ with $R(\infty)R_s(\infty)=0$ and $R_s(\infty)\geq 0$, internal stability is achieved if and only if the DC loop gain of the interconnection is strictly less than unity; i.e., $\lambda_{max}(R(0)R_s(0))<1$ (see \cite{lanzon2008stability}).

Feedback stabilization problems have been addressed in many papers using the  PR systems theory (see \cite{kokotovic1989positive,saberi1990global}, etc). In these papers, a system with a specified nonlinearity is stabilized by a state feedback control law that renders the linear part of the system PR. The essence of these papers is deriving conditions for such PR state feedback equivalence, based on which stabilization results can then be achieved. For example, \cite{saberi1990global} renders a linear system PR and this result is then generalized to nonlinear systems in \cite{byrnes1991passivity} using passivity theory. Further nonlinear generalizations of these ideas are presented in the papers \cite{byrnes1991asymptotic,santosuosso1997passivity,lin1995feedback,jiang1996passification}. However, because of the limitations of passivity and PR systems theory, the systems investigated in these papers are only allowed to have relative degree one. This rules out a wide variety of control systems with relative degree two, such as mechanical systems with force actuators and position sensors. Therefore, we seek to solve the problem of state feedback equivalence to an NI system, for systems of relative degree one or two, to complement the existing results that are based on passivity and PR systems theory.

In this paper, we investigate the NI  state feedback equivalence problem for systems of relative degree one and relative degree two. A system with no zero at the origin and of relative degree one or two can be made minimal and NI via the use of state feedback if and only if it is controllable and weakly minimum phase (see for example \cite{khalil2002nonlinear} and \cite{isidori2013nonlinear} for details about feedback linearization). In particular, a controllable system of relative degree one is state feedback equivalent to a strongly strict negative imaginary (SSNI) system if and only if it is minimum phase. The proposed NI state feedback equivalence results are then applied to a robust stabilization problem for an uncertain system with SNI uncertainty.

In addition to complementing the existing feedback equivalence results to allow for relative degree two, the present research provides alternative feedback equivalence results for systems of relative degree one based on NI systems theory, as NI systems arise naturally in a wide variety of applications \cite{petersen2016negative}. This work enables NI systems theory to be applied to a broader class of systems of relative degree one or two when full state information is available.

This paper is organised as follows: Section \ref{section:pre} provides the essential background on NI systems theory. Section \ref{section:problem statement} defines the class of systems under consideration and states the objectives of this paper. The problem is separated into the relative degree one and relative degree two cases. We present in Lemmas \ref{thm:rd1 feedback NI} and \ref{thm:rd2 feedback NI} necessary and sufficient conditions under which there exist state feedback matrices that render the system NI. Formulas for the required state feedback matrices are provided in the proofs. Also, for the special cases when the internal dynamics have zero dimension, we show that there always exist state feedback matrices that make the system NI. The main result of this paper is presented in Theorem \ref{thm:main}, which combines the NI feedback equivalence results of both the relative degree one and relative degree two cases. In Section \ref{section:SSNI}, an SSNI state feedback equivalence result is also provided for systems of relative degree one while it is explained that systems with relative degree two can never be rendered SSNI via state feedback. Section \ref{section:synthesis} applies the NI state feedback equivalence results presented in Section \ref{section:problem statement} in stabilizing an uncertain system with SNI uncertainty. Section \ref{section:example} illustrates the presented results with a numerical example. Section \ref{section:conclusion} concludes the paper and discusses possible future work.
 

%% file: Notation.tex
\textbf{Notation}: The notation in this paper is standard. $\mathbb R$ and $\mathbb C$ denote the fields of real and complex numbers, respectively. $j\mathbb R$ denotes the set of purely imaginary numbers. $\mathbb R^{m\times n}$ and $\mathbb C^{m\times n}$ denote the spaces of real and complex matrices of dimension $m\times n$, respectively. $\mathfrak{R} [\cdot]$ is the real part of a complex number. $A^T$ denotes the transpose of a matrix $A$. $A^{-T}$ denotes the transpose of the inverse of $A$; i.e., $A^{-T}=(A^{-1})^T=(A^T)^{-1}$. $\det (A)$ denotes the determinant of $A$. $ker(A)$ denotes the kernel of $A$. $spec(A)$ denotes the spectrum of $A$. $\lambda_{max}(A)$ denotes the largest eigenvalue of a matrix $A$ with real spectrum. For a symmetric matrix $P$, $P>0\ (P\geq 0)$ denotes the fact that the matrix $P$ is positive definite (positive semidefinite) and $P<0\ (P\leq 0)$ denotes the fact that the matrix $P$ is negative definite (negative semidefinite). For a positive definite matrix $P$, we denote by $P^{\frac{1}{2}}$ the unique positive definite square root of $P$. $OLHP$ and $CLHP$ are the open and closed left half-planes of the complex plane, respectively.

%% file: Preliminaries.tex
\section{PRELIMINARIES}\label{section:pre}

\begin{definition}(Negative Imaginary Systems)\cite{xiong2010negative}
A square real-rational proper transfer function matrix $R(s)$ is said to be negative imaginary (NI) if:

1. $R(s)$ has no poles at the origin and in $\mathfrak R[s]>0$;

2. $j[R(j\omega)-R^*(j\omega)]\geq 0$ for all $\omega \in (0,\infty)$ except for values of $\omega$ where $j\omega$ is a pole of $R(s)$;

3. if $j\omega_0$ with $\omega_0\in (0,\infty)$ is a pole of $R(s)$, then it is a simple pole and the residue matrix $K_0=\lim_{s\to j\omega_0}(s-j\omega_0)jR(s)$ is Hermitian and positive semidefinite.
\end{definition}

\begin{definition}(Strictly Negative Imaginary Systems)\cite{xiong2010negative}
A square real-rational proper transfer function $R(s)$ is said to be strictly negative imaginary (SNI) if the following conditions are satisfied:

1. $R(s)$ has no poles in $\mathfrak R[s]\geq 0$;

2. $j[R(j\omega)-R^*(j\omega)]> 0$ for all $\omega \in (0,\infty)$.	
\end{definition}

\begin{definition}(Strongly Strictly Negative Imaginary Systems)\cite{lanzon2011strongly}
A square real-rational proper transfer function matrix $R(s)$ is said to be strongly strictly negative imaginary (SSNI) if the following conditions are satisfied:

1. $R(s)$ is SNI.

2. $\lim_{\omega \to \infty}j\omega[R(j\omega)-R^*(j\omega)]>0$ and $\lim_{\omega \to 0}j\frac{1}{\omega}[R(j\omega)-R^*(j\omega)]>0$.
\end{definition}

\begin{lemma}(NI Lemma)\cite{xiong2010negative}\label{lemma:NI}
Let $(A,B,C,D)$ be a minimal state-space realisation of an $p\times p$ real-rational proper transfer function matrix $R(s)$ where $A\in \mathbb R^{n\times n}$, $B\in \mathbb R^{n\times p}$, $C\in \mathbb R^{p\times n}$, $D\in \mathbb R^{p\times p}$. Then $R(s)$ is NI if and only if:

1. $\det(A)\neq 0$, $D=D^T$;

2. There exists a matrix $Y=Y^T>0$, $Y\in \mathbb R^{n\times n}$ such that
\begin{equation}\label{eq:NI condition}
	AY+YA^T\leq 0,\qquad \textnormal{and} \qquad B+AYC^T=0.
\end{equation}	
\end{lemma}

\begin{lemma}(SSNI Lemma)\cite{lanzon2011strongly}\label{thm:SSNI}
Given a square transfer function matrix $R(s)\in \mathbb R^{p\times p}$ with a state-space realisation $(A,B,C,D)$, where $A\in \mathbb R^{n\times n}$, $B\in \mathbb R^{n\times p}$, $C \in \mathbb R^{p\times n}$ and $D\in \mathbb R^{p\times p}$. Suppose $R(s)+R(-s)^T$ has normal rank $p$ and $(A,B,C,D)$ has no observable uncontrollable modes. Then $A$ is Hurwitz and $R(s)$ is SSNI if and only if $D=D^T$ and there exists a matrix $Y=Y^T>0$ that satisfies conditions
\begin{equation}\label{eq:SSNI condition}
AY+YA^T<0, \quad \textnormal{and} \quad B+AYC^T=0.
\end{equation}
\end{lemma}

\begin{lemma}(Internal Stability of Interconnected NI Systems)\cite{xiong2010negative}\label{lemma:dc gain theorem}
	Consider an NI transfer function matrix $R(s)$ and an SNI transfer function matrix $R_s(s)$ that satisfy $R(\infty)R_s(\infty)=0$ and $R_s(\infty)\geq 0$. Then the positive feedback interconnection $[R(s),R_s(s)]$ is internally stable if and only if $\lambda_{max}(R(0)R_s(0))<1$.
\end{lemma}

\begin{definition}(Lyapunov Stability)\cite{bernstein2009matrix}
\label{def:LS}
A square matrix $A$ is said to be Lyapunov stable if $spec(A)\subset CLHP$ and every purely imaginary eigenvalue of $A$ is semisimple.
\end{definition}

\begin{lemma}(Lyapunov Stability Theorem - Asymptotic Stablity)\cite{hespanha2018linear}\label{thm:Lyapunov}
Consider a continuous-time homogeneous linear time-invariant (LTI) system
\begin{equation}\label{eq:LTI system}
\dot x = \mc A x,\qquad x \in \mathbb R^n	,
\end{equation}
the following statements are equivalent:

1. The system (\ref{eq:LTI system}) is asymptotically stable.

2. All the eigenvalues of $\mc A$ have strictly negative real parts.

3. For every symmetric positive definite matrix $\mc Q$, there exists a unique solution $\mc P$ to the following Lyapunov equation
\begin{equation}\label{eq:Lyapunov}
	\mc A^T\mc P+\mc P\mc A = -\mc Q
\end{equation}
Moreover, $\mc P$ is symmetric and positive definite.

4. There exists a symmetric positive definite matrix $\mc P$ for which the following Lyapunov matrix inequality holds:
\begin{equation*}
	\mc A^T\mc P+\mc P\mc A<0.
\end{equation*}
\end{lemma}

\begin{lemma}(Lyapunov Stability Theorem - Lyapunov Stable)\cite{bernstein2009matrix}\label{thm:marginally stable}
	Let $\mc A\in \mathbb R^{n\times n}$ and assume there exists a positive semidefinite matrix $\mc Q \in \mathbb R^{n\times n}$ and a positive definite matrix $\mc P \in \mathbb R^{n\times n}$ such that (\ref{eq:Lyapunov}) is satisfied, then $\mc A$ is Lyapunov stable.
\end{lemma}

\begin{lemma}(Eigenvector Test for Controllability)\cite{hespanha2018linear}\label{thm:eigenvector test for ctrl}
The pair $(A,B)$ is controllable if and only if there is no eigenvector of $A^T$ in the kernel of $B^T$.
\end{lemma}

\begin{lemma}(Eigenvector Test for Observability)\cite{hespanha2018linear}\label{thm:eigenvector test for obsv}
The pair $(A,C)$ is observable if and only if no eigenvector of $A$ is in the kernel of $C$.
\end{lemma}

%% file: Feedback_equivalence_nonlinear_NI.tex
\section{STATE FEEDBACK EQUIVALENCE TO A NEGATIVE IMAGINARY SYSTEM}
\label{section:problem statement}
Consider a linear multiple-input multiple-output (MIMO) system with the following state-space model
\begin{subequations}\label{eq:original state-space}
\begin{align}
\dot x =&\ \mc Ax+ \mc Bu,\\
y=&\ \mc Cx.
\end{align}
\end{subequations}
where $x\in \mathbb R^{n}$ and $u,y\in \mathbb R^p$. We consider the necessary and sufficient conditions under which the system (\ref{eq:original state-space}) is state feedback equivalent to an NI system. State feedback equivalence to an NI system is defined as follows:

\begin{definition}
The system (\ref{eq:original state-space}) is said to be state feedback equivalent to an NI system if there exists a state feedback control law
\begin{equation*}
u=K_x x+ K_v v	
\end{equation*}
such that the closed-loop system with the new input $v\in \mathbb R^p$ is minimal and NI.
\end{definition}

Let us recall the definition of relative degree.

\begin{definition}(Relative Degree)\cite{saberi1990global}
The system (\ref{eq:original state-space}) is said to have relative degree $r$ if its first $r-1$ Markov parameters are zero, i.e., $\mc C \mc A^i \mc B=0$ for $i=0,1,\cdots,r-2$; and $\mc C \mc A^{r-1} \mc B$ is nonsingular.	
\end{definition}

We discuss the relative degree one and two cases in the following two subsections.

\subsection{Relative Degree One Case}
\label{section:rd1}
Suppose the system (\ref{eq:original state-space}) has relative degree one, that is, $\det(\mc C \mc B)\neq 0$. Then $m:=n-p \geq 0$ and hence, without loss of generality, the system (\ref{eq:original state-space}) can be considered to be in the special coordinate basis (SCB) (see \cite{sannuti1987special,chen2004linear})
\begin{subequations}\label{eq:linearised1}
\begin{align}
\dot z =&\ A_{11}z+A_{12}y,\label{eq:linearised1a}\\
\dot y =&\ A_{21}z+A_{22}y+\mc C \mc Bu,\label{eq:linearised1b}\\
y=& \left[\begin{matrix}0& {I}\end{matrix}\right]\left[\begin{matrix}z \\ y\end{matrix}\right],\label{eq:linearised1c}
\end{align}
\end{subequations}
where $z\in \mathbb R^m$. This can be realized using a state transformation, according to Lemma \ref{lemma:rd1 state transformation}.
\begin{lemma}\label{lemma:rd1 state transformation}
	Suppose the system (\ref{eq:original state-space}) has relative degree one, that is, $\det(\mc C \mc B)\neq 0$. Then there exists a state transformation such that the resulting transformed system is of the form (\ref{eq:linearised1}).
\end{lemma}
\begin{proof}
Since $\det(\mc C \mc B) \neq 0$, then $n\geq p$, $\mc C$ has full row rank and $\mc B$ has full column rank. Let $y=\mc C x$ be a new state. However, we also need a complementary state of dimension $n-p$. Let $z= \mc C_zx$, where $\mc C_z \in \mathbb R^{(n-p)\times n}$ is such that $T=\left[\begin{matrix}\mc C_z \\ \mc C\end{matrix}\right]$ is nonsingular and $\mc C_z \mc B = 0$. The state transformation $\left[\begin{matrix}z \\ y\end{matrix}\right]=Tx$ transforms the system (\ref{eq:original state-space}) into the state-space model
\begin{align*}
	\frac{d}{dt}\left[\begin{matrix}z \\ y\end{matrix}\right]=& T\mc AT^{-1}\left[\begin{matrix}z \\ y\end{matrix}\right]+\left[\begin{matrix}0 \\ \mc C \mc B\end{matrix}\right]u,\\
	y=& \left[\begin{matrix}0& {I}\end{matrix}\right]\left[\begin{matrix}z \\ y\end{matrix}\right].
\end{align*}
Expressing the matrix $T\mc A T^{-1}$ with arbitrary block matrices, we obtain the state-space model (\ref{eq:linearised1}). This completes the proof.
\end{proof}

For the system (\ref{eq:linearised1}), since $\det(\mc C \mc B)\neq 0$, then the input $u$ can be represented as
\begin{equation*}
u=(\mc C \mc B)^{-1}\left(v+(K_1-A_{21})z+(K_2-A_{22})y\right),
\end{equation*}
and the system (\ref{eq:linearised1}) takes the form
\begin{subequations}\label{eq:linearised2}
\begin{align}
\dot z =&\ A_{11}z+A_{12}y,\label{eq:linearised2a}\\
\dot y =&\ K_1z+K_2y+v,\label{eq:linearised2b}\\
y=&\ [\begin{matrix}0& I\end{matrix}]\left[\begin{matrix}z \\ y\end{matrix}\right].\label{eq:linearised2c}
\end{align}
\end{subequations}
We need to find the state feedback matrices $K_1\in \mathbb R^{p\times m}$ and $K_2 \in \mathbb R^{p\times p}$ such that the system (\ref{eq:linearised2}) is NI. We show in the following lemma the necessary and sufficient conditions under which such state feedback matrices exist.

\begin{lemma}\label{thm:rd1 feedback NI}
Suppose the system (\ref{eq:linearised2}) satisfies $\det (A_{11}) \neq 0$. Then there exist $K_1\in \mathbb R^{p\times m}$ and $K_2\in \mathbb R^{p\times p}$ such that the system (\ref{eq:linearised2}) is an NI system with minimal realisation if and only if the pair $(A_{11},A_{12})$ is controllable and $A_{11}$ is Lyapunov stable.
\end{lemma}
\begin{proof}
Let us define the following:
\begin{align}
A=&	\left[\begin{matrix}A_{11}&A_{12}\\K_1&K_2\end{matrix}\right],\label{eq:rd1 A}\\
B=&\left[\begin{matrix}0\\ I\end{matrix}\right],\label{eq:rd1 B}\\
C=& \left[\begin{matrix}0& I\end{matrix}\right].\label{eq:rd1 C}
\end{align}
First, we prove that the controllabilities of the pairs $(A,B)$ and $(A_{11},A_{12})$ are equivalent. According to Lemma \ref{thm:eigenvector test for ctrl}, the pair $(A,B)$ is controllable if and only if any vector in the kernel of $B^T$ is not an eigenvector of $A^T$. Due to the particular form of $B$, this means for $\small\zeta=\left[\begin{matrix}\zeta_1\\0\end{matrix}\right]$, where $\zeta_1\in \mathbb R^{m}$ is any nonzero vector, there does not exist a scalar $\lambda_c$ such that $A^T\zeta=\lambda_c \zeta$. That is $\small\left[\begin{matrix}A_{11}^T\zeta_1 \\ A_{12}^T\zeta_1\end{matrix}\right]\neq\left[\begin{matrix}\lambda_c\zeta_1 \\ 0\end{matrix}\right]$ for any $\zeta_1 \neq 0$ and $\lambda_c$. This means a nonzero vector $\zeta_1$ cannot be an eigenvector of $A_{11}^T$ and in $ker(A_{12}^T)$ at the same time, which is true if and only if $(A_{11},A_{12})$ is controllable.

\textbf{Sufficiency}. We now prove that the observabilities of the pairs $(A_{11},K_1)$ and $(A,C)$ are equivalent. According to Lemma \ref{thm:eigenvector test for obsv}, $(A,C)$ is observable if and only if any vector in the kernel of $C$ is not an eigenvector of $A$. Due to the particular form of $C$, this means for any vector $\small \phi=\left[\begin{matrix}\phi_1\\0\end{matrix}\right]$, where $\phi_1\in\mathbb R^{m}$ is any nonzero vector, there does not exist a scalar $\lambda_o$ such that $A\phi=\lambda_o\phi$. That is $\small\left[\begin{matrix}A_{11}\phi_1 \\ K_1\phi_1\end{matrix}\right]\neq\left[\begin{matrix}\lambda_o\phi_1 \\ 0\end{matrix}\right]$ for any $\phi_1$ and $\lambda_o$. This means no eigenvector of $A_{11}$ is in $ker(K_1)$, which is true if and only if $(A_{11},K_1)$ is observable.

The nonsingular matrix $A_{11}$ is Lyapunov stable (see Definition \ref{def:LS}) if and only if there exists a state transformation $A_{11} \mapsto SA_{11}S^{-1}$ which allows $A_{11}$ to be represented, without loss of generality, as $A_{11}=diag(A_{11}^a,A_{11}^b)$, where
\begin{equation}\label{eq:A_11a A_11b definition}
\begin{aligned}
	& spec( A_{11}^a)\subset j\mathbb R\backslash \{0\}, \quad spec( A_{11}^b)\subset OLHP,\\
	& \textnormal{and} \quad A_{11}^a+( A_{11}^a)^T=0.
\end{aligned}
\end{equation}
Here $A_{11}^a \in \mathbb R^{m_a\times m_a}$ and $A_{11}^b \in \mathbb R^{m_b\times m_b}$, where $m_a \geq 0$, $m_b \geq 0$ and $m_a+m_b=m$. The conditions in (\ref{eq:A_11a A_11b definition}) are achievable according to the proof of Proposition 11.9.6 in \cite{bernstein2009matrix}. Decomposing $A_{12}$ and $K_1$ accordingly using the same state-space transformation, we can write (\ref{eq:linearised2}) as
\begin{align*}
\dot z_1 =&\ A_{11}^az_1+A_{12}^a y,\\
\dot z_2 =&\ A_{11}^bz_2+A_{12}^b y,\\
\dot y =&\ K_1^az_1 + K_1^bz_2+K_2y+v,\\
y=& \left[\begin{matrix}0 & 0& I\end{matrix}\right]\left[\begin{matrix}z_1 \\ z_2 \\ y\end{matrix}\right].
\end{align*}
Since $A_{11}^b$ is Hurwitz, there exist $\mc Y_1^b=(\mc Y_1^b)^T>0$ and $Q_b=Q_b^T>0$ such that
\begin{equation*}
	A_{11}^b\mc Y_1^b+\mc Y_1^b (A_{11}^b)^T=-Q_b.
\end{equation*}
Let $K_1$ be defined as
\begin{equation}\label{eq:rd1 K_1}
	K_1= \left[\begin{matrix}K_1^a & K_1^b\end{matrix}\right],
\end{equation}
where
\begin{equation*}
K_1^a=-{A_{12}^a}^T(A_{11}^a)^{-T},
\end{equation*}
and
\begin{equation}\label{eq:K1b}
K_1^b=\left(-{A_{12}^b}^T(A_{11}^b)^{-T}+\mc H_b\right) (\mc Y_1^b)^{-1}.
\end{equation}
Here, $\mc H_b$ is contained in the set
\begin{equation}\label{eq:S_1}
S_1=\{\mc H_b\in \mathbb R^{p\times m_b}:\mc H_b^T\mc H_b \leq 2 Q_b\}.
\end{equation}
If $(A_{11}, A_{12})$ is controllable, we can always find $\mc H_b$ such that $(A_{11},K_1)$ observable. We prove this in the following. According to Lemma \ref{thm:eigenvector test for ctrl}, the controllability of $(A_{11}, A_{12})$ implies that no eigenvector of $diag\left((A_{11}^a)^T,(A_{11}^b)^T\right)$ is in $ker\left(\left[\begin{matrix}(A_{12}^a)^T & (A_{12}^b)^T\end{matrix}\right]\right)$. This implies that both $(A_{11}^a,A_{12}^a)$ and $(A_{11}^b,A_{12}^b)$ are controllable, which can be proved by applying the eigenvector tests in Lemma \ref{thm:eigenvector test for ctrl} to the vectors $\small\left[\begin{matrix}\zeta_a\\ 0\end{matrix}\right]$ and $\small\left[\begin{matrix}0\\ \zeta_b\end{matrix}\right]$, where $\zeta_a$ and $\zeta_b$ are eigenvectors of $(A_{11}^a)^T$ and $(A_{11}^b)^T$, respectively. According to Lemma \ref{thm:eigenvector test for obsv}, $(A_{11},K_1)$ is observable if and only if for any nonzero vector $\small\phi_K = \left[\begin{matrix}\phi_a \\ \phi_b\end{matrix}\right]$, which is an eigenvector of $A_{11}$, we have $K_1\phi_K\neq 0$. Since $A_{11}^a$ and $A_{11}^b$ have no common eigenvalue, then $\phi_K$ is an eigenvector of $A_{11}$ only if $\phi_a = 0$ or $\phi_b=0$. We consider two cases:

\textbf{Case 1}. $\phi_a \neq 0$ and $\phi_b = 0$. In this case, $\phi_a$ is an eigenvector of $A_{11}^a$; i.e., $A_{11}^a \phi_a = \lambda_a \phi_a$. Since $A_{11}^a+(A_{11}^a)^T=0$, we have $(A_{11}^a)^T\phi_a = -\lambda_a \phi_a$. Hence, $(A_{11}^a)^{-T}\phi_a = -\frac{1}{\lambda_a} \phi_a$. Also, because $(A_{11}^a,A_{12}^a)$ is controllable, $(A_{12}^a)^T\phi_a\neq 0$. Therefore,
\begin{equation*}
K_1\phi_K=K_1^a\phi_a=-{A_{12}^a}^T(A_{11}^a)^{-T} \phi_a = \frac{1}{\lambda_a}{A_{12}^a}^T	\phi_a \neq 0.
\end{equation*}

\textbf{Case 2}. $\phi_a = 0$ and $\phi_b \neq 0$. In this case, $\phi_b$ is an eigenvector of $A_{11}^b$. Because $-{A_{12}^b}^T(A_{11}^b)^{-T}$ in (\ref{eq:K1b}) is fixed, and $S_1$ has nonempty interior due to the positive definiteness of $Q_b$, then we can always find $\mc H_b$ such that
\begin{align*}
K_1\phi_K=&K_1^b\phi_b \\
=&\left(-{A_{12}^b}^T\left(A_{11}^b\right)^{-T}+ \mc H_b\right) (\mc Y_1^b)^{-1}\phi_b \neq 0,
\end{align*}
for all $\phi_b$ that are eigenvectors of $A_{11}^b$. Therefore, with this particular choice of $\mc H_b$, we have $(A_{11},K_1)$ observable. In this case, $(A,C)$ is also observable. The controllability of $(A_{11},A_{12})$ implies that the realisation $(A,B,C)$ in (\ref{eq:rd1 A}), (\ref{eq:rd1 B}) and (\ref{eq:rd1 C}) is minimal.

Let $K_2$ be defined as
\begin{equation}\label{eq:K_2 definition}
K_2 = K_1A_{11}^{-1}A_{12}-\mc Y_2^{-1},
\end{equation}
where $\mc Y_2 \in \mathbb R^{p\times p}$ can be any symmetric positive definite matrix; i.e., $\mc Y_2=\mc Y_2^T >0$. We apply Lemma \ref{lemma:NI} in the following to prove that the system (\ref{eq:linearised2}) is an NI system. We construct the matrix $Y$ as follows:
\begin{equation}\label{eq:Y}
Y=\left[\begin{matrix}\mc Y_1 + A_{11}^{-1}A_{12}\mc Y_2 A_{12}^TA_{11}^{-T} & -A_{11}^{-1}A_{12}\mc Y_2\\-\mc Y_2 A_{12}^TA_{11}^{-T} & \mc Y_2\end{matrix}\right],
\end{equation}
where $\mc Y_1=diag(y_1^a I,\mc Y_1^b)$ with $y_1^a>0$ being a scalar. We have $Y>0$ because $\mc Y_2>0$ and the Schur complement of the block $\mc Y_2$ is
\begin{align*}
	Y/& \mc Y_2\notag\\
	 =& \mc Y_1 + A_{11}^{-1}A_{12}\mc Y_2 A_{12}^TA_{11}^{-T}-A_{11}^{-1}A_{12}\mc Y_2{\mc Y_2}^{-1}\mc Y_2 A_{12}^TA_{11}^{-T}\notag\\
	 =& \mc Y_1 >0.
\end{align*}
For Condition 1 in Lemma \ref{lemma:NI}, the determinant of the matrix $A$ in (\ref{eq:rd1 A}) is
\begin{align*}
	\det{(A)}=& \det{(A_{11})}\det{(K_2-K_1A_{11}^{-1}A_{12})}\notag\\
	=&\det{(A_{11})}\det{(-\mc Y_2^{-1})}.
\end{align*}
Because $\det(A_{11})\neq 0$ and $\mc Y_2>0$, $\det {(A)}\neq 0$. Also, there is no input feedthrough in the output equation (\ref{eq:linearised2c}). Hence, Condition 1 in Lemma \ref{lemma:NI} is satisfied. For Condition 2 in Lemma \ref{lemma:NI}, with $Y$ defined in (\ref{eq:Y}), we have
\begin{equation*}
AY=	\left[\begin{matrix}y_1^a A_{11}^a & 0 & 0\\ 0 & A_{11}^b\mc Y_1^b & 0 \\0 & \mc H_b & -I\end{matrix}\right].
\end{equation*}
Therefore, $AYC^T=-B$ and
\begin{equation*}
AY+YA^T =\left[\begin{matrix}0 & 0 & 0\\ 0 & -Q_b &  \mc H_b^T \\0 &  \mc H_b & -2I\end{matrix}\right].
\end{equation*}
For the matrix $\left[\begin{matrix} Q_b & -\mc H_b^T \\ -\mc H_b & 2I\end{matrix}\right]$, we have $2I>0$ and the Schur complement of the block $2I$ is
\begin{align*}
	Q_b-\frac{1}{2}\mc H_b^TH_b \geq 0,
\end{align*}
where (\ref{eq:S_1}) is also used. Therefore, $AY+YA^T\leq 0$. Condition 2 in Lemma \ref{lemma:NI} is satisfied. Hence, the system (\ref{eq:linearised2}) is an NI system with minimal realisation.

\textbf{Necessity}. If the realisation $(A,B,C)$ is minimal and the system (\ref{eq:linearised2}) is NI, then according to the proof of Lemma \ref{lemma:NI} (see Lemma 7 in \cite{xiong2010negative}), there exists an $X=X^T > 0$ such that
\begin{equation*}\label{eq:rd1 LMI}
	\left[\begin{matrix} XA+A^TX & XB-A^TC^T \\ B^TX-CA & -(CB+B^TC^T)\end{matrix}\right] \leq 0.
\end{equation*}
Therefore, for any $z$, $y$ and $v$, we have
\begin{equation}\label{eq:rd1 LMI z y v}
	\left[\begin{matrix} z \\ y \\ v\end{matrix}\right]^T\left[\begin{matrix} XA+A^TX & XB-A^TC^T \\ B^TX-CA & -(CB+B^TC^T)\end{matrix}\right]\left[\begin{matrix} z \\ y \\ v\end{matrix}\right] \leq 0.
\end{equation}
Let $\small X=\left[\begin{matrix} X_{11} & X_{12}\\ X_{12}^T & X_{22}\end{matrix}\right]$ and substitute (\ref{eq:rd1 A}), (\ref{eq:rd1 B}) and (\ref{eq:rd1 C}) into (\ref{eq:rd1 LMI z y v}). Also, let $y=0$ and $v=-K_1z$. We have that
\begin{equation*}
	z^T(X_{11}A_{11}+A_{11}^TX_{11})z \leq 0
\end{equation*}
for any $z$, which implies that $X_{11}A_{11}+A_{11}^TX_{11}\leq 0$. Considering $X=X^T>0$, we have $X_{11}>0$. Also, since $\det(A_{11})\neq 0$, then according to Lemma \ref{thm:marginally stable}, $A_{11}$ is Lyapunov stable. Also, the controllability of $(A,B)$ implies the controllability of $(A_{11},A_{12})$. This completes the proof.
\end{proof}

\begin{remark}
The inverse of the matrix $Y$ in (\ref{eq:Y}) is
\begin{equation*}
P=Y^{-1}=\left[\begin{matrix} \mc P_1 & \mc P_1 A_{11}^{-1}A_{12} \\ A_{12}^TA_{11}^{-T}\mc P_1 & \mc P_2 + A_{12}^TA_{11}^{-T}\mc P_1 A_{11}^{-1}A_{12}\end{matrix}\right],
\end{equation*}
where $\mc P_1=\mc Y_1^{-1}>0$ and $\mc P_2=\mc Y_2^{-1}>0$. This enable us to define a storage function for the system (\ref{eq:linearised2}) as
\begin{equation*}
V_1(z,y)=\frac{1}{2}	\left[\begin{matrix} z^T & y^T\end{matrix}\right]P\left[\begin{matrix} z \\ y\end{matrix}\right],
\end{equation*}
which satisfies
\begin{equation*}
\dot V_1(z,y) \leq v^T \dot y.	
\end{equation*}
\end{remark}

\begin{remark}
	In the case that $m=0$, the system (\ref{eq:linearised1}) becomes
\begin{subequations}\label{eq:rd1 no z}
\begin{align}
	\dot y =&\ A_{22}y+\mc C \mc B u,\\
	y=&\ y.
\end{align}	
\end{subequations}	
Let the input be
\begin{equation*}
u=(\mc C \mc B)^{-1}(v_0+(K_0-A_{22})y).
\end{equation*}
Then, the system (\ref{eq:rd1 no z}) becomes
\begin{subequations}\label{eq:rd1 no z final}
\begin{align}
	\dot y =&\ K_0 y+v_0,\\
	y=&\ y.
\end{align}
\end{subequations}
Choose $K_0$ to be such that $K_0=K_0^T<0$. The matrix $Y_0=-K_0^{-1}$ then satisfies the conditions in (\ref{eq:NI condition}) and (\ref{eq:SSNI condition}), while all other conditions in Lemma \ref{lemma:NI} and \ref{thm:SSNI} are satisfied. Therefore, there exists a $K_0$ such that the system (\ref{eq:rd1 no z final}) is an SSNI system.
\end{remark}

\subsection{Relative Degree Two Case}
\label{section:rd2}
Suppose the system (\ref{eq:original state-space}) has relative degree two, that is $\mc C \mc B=0$ and $\det(\mc C \mc A \mc B)\neq 0$. Then $m:=n-2p \geq 0$ and hence without loss of generality, the system (\ref{eq:original state-space}) can be considered to be in the SCB (see \cite{sannuti1987special,chen2004linear})
\begin{subequations}\label{eq:linearised3}
\begin{align}
\dot {z} =&\  A_{11} z+ A_{12}x_1+ A_{13}x_2,\label{eq:linearised3a}\\
\dot x_1 =&\ x_2,\label{eq:linearised3b}\\
\dot x_2 =&\  A_{31} z +  A_{32} x_1+ A_{33}x_2+\mc C\mc A\mc B u,\label{eq:linearised3c}\\
y =& \left[\begin{matrix}0& I & 0\end{matrix}\right]\left[\begin{matrix} z \\ x_1 \\ x_2\end{matrix}\right],\label{eq:linearised3d}
\end{align}
\end{subequations}
where $ z\in \mathbb R^{m}$ and $x_1,x_2 \in \mathbb R^p$. This can be realized using a state transformation, according to Lemma \ref{lemma:rd2 state transformation}.
\begin{lemma}\label{lemma:rd2 state transformation}
Suppose the system (\ref{eq:original state-space}) has relative degree two, that is, $\mc C \mc B=0$ and $\det(\mc C \mc A \mc B) \neq 0$. Then there exists a state transformation such that the resulting transformed system is of the form (\ref{eq:linearised3}).
\end{lemma}
\begin{proof}
Since $\det(\mc C \mc A \mc B) \neq 0$, then both $\mc C$ and $\mc C \mc A$ have full row rank. Let $x_1=y=\mc C x$ and $x_2=\dot x_1= \mc C \mc Ax$. Considering that $\mc C \mc B=0$ and $\mc C \mc A \mc B$ is nonsingular, we have that the matrix
$\left[\begin{matrix} \mc C \\ \mc C \mc A\end{matrix}\right]$ is full row rank. We also need a complementary state $z$ of dimension $n-2p$. Let $z=\mc C_z x$, where $C_z$ is such that $T=\left[\begin{matrix}\mc C_z \\ \mc C \\ \mc C \mc A\end{matrix}\right]$ is nonsingular and $\mc C_z \mc B=0$. Under the state transformation $\left[\begin{matrix}z \\ x_1 \\ x_2 \end{matrix}\right]=Tx$, the system (\ref{eq:original state-space}) now becomes
\begin{align*}
	\frac{d}{dt}\left[\begin{matrix}z \\ x_1 \\ x_2 \end{matrix}\right]=&\  T\mc AT^{-1}\left[\begin{matrix}z \\ x_1 \\ x_2 \end{matrix}\right] + \left[\begin{matrix}0 \\ 0 \\ \mc C \mc A \mc B \end{matrix}\right]u,\\
	y =& \left[\begin{matrix}0& I & 0\end{matrix}\right]\left[\begin{matrix} z \\ x_1 \\ x_2\end{matrix}\right].
\end{align*}
We express $T \mc A T^{-1}$ using arbitrary block matrices and we also use the fact that $\dot x_1 = x_2$. Therefore, we obtain the state-space model (\ref{eq:linearised3}). This completes the proof.
\end{proof}

Let the input of system (\ref{eq:linearised3}) be
\begin{equation}\label{eq:input}
\small u=(\mc C \mc A \mc B)^{-1}(v+(K_1-A_{31}) z+(K_2-A_{32})x_1+(K_3-A_{33})x_2),	
\end{equation}
then the system (\ref{eq:linearised3}) takes the form
\begin{subequations}\label{eq:linearised4}
\begin{align}
\dot { z} =&\  A_{11} z+ A_{12}x_1+ A_{13}x_2,\label{eq:linearised4a}\\
\dot x_1 =&\ x_2,\label{eq:linearised4b}\\
\dot x_2 =&\  K_1  z+ K_2x_1+  K_3x_2+  v,\label{eq:linearised4c}\\
y =&\ [\begin{matrix}0& I & 0\end{matrix}]\left[\begin{matrix} z \\ x_1 \\ x_2\end{matrix}\right].\label{eq:linearised4d}
\end{align}
\end{subequations}
We need to find the state feedback matrices $ K_1\in \mathbb R^{p\times  m}$, $ K_2 \in \mathbb R^{p\times p}$ and $ K_3 \in \mathbb R^{p\times p}$ such that the system (\ref{eq:linearised4}) is NI. We show in the following lemma the necessary and sufficient conditions under which such state feedback matrices exist.

\begin{lemma}\label{thm:rd2 feedback NI}
Suppose the system (\ref{eq:linearised4}) satisfies $\det ( A_{11}) \neq 0$. Then there exist $ K_1$, $ K_2$ and $K_3$ such that the system (\ref{eq:linearised4}) is an NI system with minimal realisation if and only if $ A_{11}$ is Lyapunov stable and the pair $( A_{11}, A_{11} A_{13}+ A_{12})$ is controllable.
\end{lemma}
\begin{proof}
Let us define the following:
\begin{align}
 A=&	\left[\begin{matrix} A_{11}& A_{12}& A_{13}\\ 0&0&I\\  K_1&  K_2&  K_3\end{matrix}\right],\label{eq:rd2 A}\\
 B=&\left[\begin{matrix}0\\0\\ I\end{matrix}\right],\label{eq:rd2 B}\\
 C=& \left[\begin{matrix}0 & I &0\end{matrix}\right]\label{eq:rd2 C}.
\end{align}
First, we prove that the controllabilities of the pairs $( A,  B)$ and $( A_{11}, A_{11} A_{13}+ A_{12})$ are equivalent. According to Lemma \ref{thm:eigenvector test for ctrl}, the pair $( A,  B)$ is controllable if and only if any vector in the kernel of $ B^T$ is not an eigenvector of $ A^T$. Due to the particular form of $ B$, this means for the vector $\small\eta=\left[\begin{matrix}\eta_1\\ \eta_2\\ 0\end{matrix}\right]\neq 0$ with $\eta_1\in \mathbb R^{ m}$ and $\eta_2\in \mathbb R^p$ being any vectors that are not both zero, there does not exist a scalar $ \lambda_c$ such that $ A^T \eta =  \lambda_c \eta$. That is $\small\left[\begin{matrix} A_{11}^T\eta_1\\  A_{12}^T\eta_1\\  A_{13}^T\eta_1+\eta_2\end{matrix}\right]\neq\left[\begin{matrix} \lambda_c\eta_1\\  \lambda_c\eta_2\\ 0\end{matrix}\right]$ for any $\eta_1$, $\eta_2$ and $ \lambda_c$. This means if $\eta_1$ is an eigenvector of $ A_{11}^T$ and $\eta_2=- A_{13}^T\eta_1$, we have $ A_{12}^T\eta_1 \neq  \lambda_c \eta_2=- \lambda_c  A_{13}^T\eta_1=- A_{13}^T A_{11}^T \eta_1$. That is $( A_{13}^T A_{11}^T+ A_{12}^T)\eta_1\neq 0$, which is true if and only if $( A_{11}, A_{11} A_{13}+ A_{12})$ is controllable.

\textbf{Sufficiency}. We now prove that the observabilities of the pairs $( A, C)$ and $( A_{11}, K_1)$ are equivalent. According to Lemma \ref{thm:eigenvector test for obsv}, $( A, C)$ is observable if and only if any vector in the kernel of $ C$ is not an eigenvector of $ A$. Due to the particular form of $ C$, this means for the vector $\small\delta=\left[\begin{matrix}\delta_1\\0\\ \delta_2\\ \end{matrix}\right]\neq 0$ with $\delta_1\in \mathbb R^{ m}$ and $\delta_2\in \mathbb R^{p}$ being any vectors that are not both zero, there does not exist a scalar $ \lambda_o$ such that $ A \delta =  \lambda_o \delta$. That is $\small\left[\begin{matrix} A_{11}\delta_1+ A_{13}\delta_2\\ \delta_2\\  K_1\delta_1+ K_3\delta_2\end{matrix}\right]\neq \left[\begin{matrix} \lambda_o\delta_1\\0\\  \lambda_o\delta_2 \end{matrix}\right]$ for any $\delta_1$, $\delta_2$ and $ \lambda_o$. Let $\delta_2=0$, then $\small\left[\begin{matrix} A_{11}\delta_1\\  K_1\delta_1\end{matrix}\right]\neq \left[\begin{matrix} \lambda_o\delta_1\\0\end{matrix}\right]$ for any $\delta_1\neq 0$ and $\lambda_o$. This means that any nonzero vector cannot be an eigenvector of $ A_{11}$ and in $ker( K_1)$ at the same time, which is true if and only if $( A_{11}, K_1)$ is observable.

The nonsingular matrix $A_{11}$ is Lyapunov stable (see Definition \ref{def:LS}) if and only if there exists a state transformation $A_{11} \mapsto SA_{11}S^{-1}$ which allows $A_{11}$ to be represented, without loss of generality, as $ A_{11}=diag( A_{11}^a, A_{11}^b)$, where
\begin{equation}\label{eq:rd2 A_11a A_11b definition}
\begin{aligned}
	& spec( A_{11}^a)\subset j\mathbb R\backslash \{0\}, \quad spec( A_{11}^b)\subset OLHP,\\
	& \textnormal{and} \quad A_{11}^a+( A_{11}^a)^T=0.
\end{aligned}
\end{equation}
Here $A_{11}^a \in \mathbb R^{m_a\times m_a}$ and $A_{11}^b \in \mathbb R^{m_b\times m_b}$, where $m_a \geq 0$, $m_b \geq 0$ and $m_a+m_b=m$. The conditions in (\ref{eq:rd2 A_11a A_11b definition}) are achievable according to the proof of Proposition 11.9.6 in \cite{bernstein2009matrix}. Decomposing $ A_{12}$, $ A_{13}$ and $ K_1$ accordingly using the same state-space transformation, we can write (\ref{eq:linearised4}) as
\begin{align*}
\dot { z}_1 =&\  A_{11}^a z_1+ A_{12}^a x_1+ A_{13}^a x_2,\\
\dot { z}_2 =&\  A_{11}^b z_2+ A_{12}^b x_1+ A_{13}^b x_2,\\
\dot x_1 = &\ x_2\\
\dot x_2 =&\  K_1^a  z_1 +  K_1^b  z_2+ K_2 x_2 +  K_3 x_2+  v,\\
y=&\ [\begin{matrix}0&0& I&0\end{matrix}]\left[\begin{matrix} z_1 \\  z_2 \\ x_1\\ x_2\end{matrix}\right].
\end{align*}
Since $ A_{11}^b$ is Hurwitz, there exist ${\mc Y}_1^b=({\mc Y}_1^b)^T>0$ and $ Q_b= Q_b^T>0$ such that
\begin{equation*}
	 A_{11}^b {\mc Y}_1^b+{\mc Y}_1^b ( A_{11}^b)^T=- Q_b.
\end{equation*}
Let $ K_3\in \mathbb R^{p\times p}$ be any matrix such that
\begin{equation}\label{eq:rd2 K_3}
K_3+ K_3^T<0.
\end{equation}
Let
\begin{equation}\label{eq:E}
E:=\left(-( K_3+ K_3^T)\right)^{\frac{1}{2}},	
\end{equation}
and hence $E=E^T>0$. Let $ K_1$ be defined as
\begin{equation}\label{eq:rd2 K_1}
	 K_1=\left[\begin{matrix} K_1^a &  K_1^b\end{matrix}\right].
\end{equation}
where
\begin{equation*}
K_1^a=-( A_{12}^a)^T( A_{11}^a)^{-T}-( A_{13}^a)^T,
\end{equation*}
and
\begin{equation}\label{eq:rd2 K1b}
K_1^b=\left(-( A_{12}^b)^T( A_{11}^b)^{-T}-( A_{13}^b)^T+ E{\mc H}_b\right) ({\mc Y}_1^b)^{-1}.
\end{equation}
Here, $\mc H_b$ is contained in the set
\begin{equation}\label{eq:S_2}
S_2=\{\mc H_b\in \mathbb R^{p\times m_b}:\mc H_b^T\mc H_b \leq Q_b\}.	
\end{equation}
We prove in the following that there exist $\mc H_b$ and $ K_3$ such that $(A_{11}, K_1)$ is observable. According to Lemma \ref{thm:eigenvector test for ctrl}, the controllability of $( A_{11}, A_{11} A_{13}+ A_{12})$ implies that no eigenvector of $diag\left(( A_{11}^a)^T, ( A_{11}^b)^T\right)$ is in the kernel of $\small\left[\begin{matrix}( A_{13}^a)^T( A_{11}^a)^T+( A_{12}^a)^T & ( A_{13}^b)^T( A_{11}^b)^T+( A_{12}^b)^T\end{matrix}\right]$. This implies that both $( A_{11}^a,  A_{11}^a A_{13}^a+ A_{12}^a)$ and $( A_{11}^b,  A_{11}^b A_{13}^b+ A_{12}^b)$ are controllable, which can be proved by applying the eigenvector tests in Lemma \ref{thm:eigenvector test for ctrl} to the vectors $\small\left[\begin{matrix}\eta_a \\ 0\end{matrix}\right]$ and $\small\left[\begin{matrix}0 \\ \eta_b\end{matrix}\right]$, where $\eta_a$ and $\eta_b$ are eigenvectors of $( A_{11}^a)^T$ and $( A_{11}^b)^T$, respectively. According to Lemma \ref{thm:eigenvector test for obsv}, $( A_{11}, K_1)$ is observable if and only if for any nonzero vector $\small\delta_K=\left[\begin{matrix}\delta_a \\ \delta_b\end{matrix}\right]$, which is an eigenvector of $ A_{11}$, we have $ K_1 \delta_K\neq 0$. Since $ A_{11}^a$ and $ A_{11}^b$ have no common eigenvalue, then $\delta_K$ is an eigenvector of $ A_{11}$ only if $\delta_a=0$ or $\delta_b=0$. We consider two cases:

\textbf{Case 1}. $\delta_a\neq 0$ and $\delta_b=0$. In this case, $\delta_a$ is an eigenvector of $ A_{11}^a$; i.e., $ A_{11}^a\delta_a=\lambda_a \delta_a$. Since $ A_{11}^a+( A_{11}^a)^T=0$, we have $( A_{11}^a)^T \delta_a=-\lambda_a \delta_a$. Hence, $(A_{11}^a)^{-T}\delta_a=-\frac{1}{ \lambda_a}\delta_a$. Also, because $( A_{11}^a,  A_{11}^a A_{13}^a+ A_{12}^a)$ is controllable, $\left(( A_{13}^a)^T( A_{11}^a)^T+( A_{12}^a)^T\right)\delta_a\neq 0$. Therefore,
\begin{align*}
	 K_1\delta_K= K_1^a\delta_a=& \left(-( A_{12}^a)^T( A_{11}^a)^{-T}-( A_{13}^a)^T\right)\delta_a\notag\\
	=&-\left(( A_{13}^a)^T( A_{11}^a)^T+( A_{12}^a)^T\right)(A_{11}^a)^{-T}\delta_a\notag\\
	=&\ \frac{1}{ \lambda_a}\left(( A_{13}^a)^T( A_{11}^a)^T+( A_{12}^a)^T\right)\delta_a\neq 0.
\end{align*}

\textbf{Case 2}. $\delta_a=0$ and $\delta_b \neq 0$. In this case, $\delta_b$ is an eigenvector of $A_{11}^b$.
Because $-( A_{12}^b)^T( A_{11}^b)^{-T}-( A_{13}^b)^T$ in (\ref{eq:rd2 K1b}) is fixed, $ E >0$ and $S_2$ has nonempty interior due to the positive definiteness of $Q_b$, then for any $K_3$, we can always find $\mc H_b$ such that
\begin{align*}\label{eq:rd2 case 2}
 K_1\delta_K =&\  K_1^b\delta_b \notag\\
 =& \left(-( A_{12}^b)^T( A_{11}^b)^{-T}-( A_{13}^b)^T+ E{\mc H}_b\right)({\mc Y}_1^b)^{-1}\delta_b\notag\\
\neq &\ 0,
\end{align*}
for all $\delta_b$ that are eigenvalues of $A_{11}^b$. We conclude that there exists $\mc H_b$ and $ K_3$ such that $( A_{11}, K_1)$ is observable. We choose such $ K_3$ and $ \mc H_b$ for the following proof.

In this case, $(A, C)$ is also observable. The controllability of $( A_{11}, A_{11} A_{13}+ A_{12})$ implies that the realisation $( A, B, C)$ in (\ref{eq:rd2 A}), (\ref{eq:rd2 B}) and (\ref{eq:rd2 C}) is minimal. Let $ K_2$ be defined as
\begin{equation}\label{eq:rd2 K_2 definition}
 K_2 =  K_1 A_{11}^{-1} A_{12}-{\mc Y}_2^{-1},
\end{equation}
where ${\mc Y}_2 \in \mathbb R^{p\times p}$ can be any symmetric positive definite matrix; i.e., ${\mc Y}_2={\mc Y}_2^T >0$. Now, we apply Lemma \ref{lemma:NI} to prove that the system (\ref{eq:linearised4}) is an NI system. First, we construct the matrix $ Y$ as
\begin{equation}\label{eq:Y rd2}
 Y=\left[\begin{matrix}{\mc Y}_1+ A_{11}^{-1} A_{12}{\mc Y}_2  A_{12}^T  A_{11}^{-T} & - A_{11}^{-1} A_{12}{\mc Y}_2 & 0 \\ -{\mc Y}_2  A_{12}^T A_{11}^{-T} & {\mc Y}_2 & 0\\ 0 & 0 & I \end{matrix}\right],
\end{equation}
where $ {\mc Y}_1=diag(y_1^a I, {\mc Y}_1^b)$ with $y_1^a>0$ being a scalar. The matrix $ Y$ is block diagonal and the first $(m+p)\times (m+p)$ diagonal block is positive definite because $ {\mc Y}_2>0$ and the Schur complement of ${\mc Y}_2$ of the first diagonal block of (\ref{eq:Y rd2}) is $ {\mc Y}_1>0$. The other diagonal block is $I>0$. Therefore, $ Y>0$. For Condition 1 in Lemma \ref{lemma:NI}, we have
\begin{align*}
	\det{(A)}
	=& -\det{(\left[\begin{matrix} A_{11}& A_{12}& A_{13}\\  K_1&  K_2&  K_3\\ 0&0&I\end{matrix}\right])}\notag \\
	=&-\det{(\left[\begin{matrix} A_{11}& A_{12}\\  K_1&  K_2 \end{matrix}\right])}\notag\\
	=&-\det{(A_{11})}\det{(K_2-K_1A_{11}^{-1}A_{12})}\notag\\
	=&\det{(A_{11})}\det{(\mc Y_2^{-1})}\notag\\
	\neq & \ 0,
\end{align*}
where the last equality uses (\ref{eq:rd2 K_2 definition}). Also, there is no input feedthrough term in the output equation (\ref{eq:linearised4d}), the second equality in Condition 1 of Lemma \ref{lemma:NI} is also satisfied. For Condition 2 in Lemma \ref{lemma:NI}, we have
\begin{equation*}
	 A Y =\left[\begin{matrix} y_1^a A_{11}^a & 0 & 0 &  A_{13}^a\\ 0 &  A_{11}^b{\mc Y}_1^b &0 & A_{13}^b \\ 0 & 0 & 0 & I\\ -( A_{13}^a)^T & -( A_{13}^b)^T+ E{\mc H}_b & -I &  K_3\end{matrix}\right].
\end{equation*}
Therefore, $ A  Y  C^T=- B$ and
\begin{equation*}
	 A Y+ Y  A^T =\left[\begin{matrix}0 & 0 & 0 & 0\\ 0 & - Q_b &0 & {\mc H}_b^T E \\ 0 & 0 & 0 & 0\\ 0 &  E{\mc H}_b & 0 &  K_3+ K_3^T\end{matrix}\right].
\end{equation*}
For the matrix $\left[\begin{matrix} Q_b & -{\mc H}_b^T E \\ -E{\mc H}_b &  -(K_3+ K_3^T)\end{matrix}\right]$, we have $-(K_3+ K_3^T)>0$ and the Schur complement of the block $-(K_3+ K_3^T)>0$ is
\begin{equation*}
Q_b - \mc H_b^T E(-(K_3+ K_3^T))^{-1}E\mc H_b=Q_b -\mc H_b^T \mc H_b \geq 0,
\end{equation*}
where (\ref{eq:E}) and (\ref{eq:S_2}) are also used. Therefore, $ A Y+ Y A^T\leq 0$. Condition 2 in Lemma \ref{lemma:NI} is also satisfied. Hence, the system (\ref{eq:linearised4}) is an NI system.

\textbf{Necessity}. Because the realisation $(A,B,C)$ is minimal and the system (\ref{eq:linearised4}) is NI, then according to the proof of Lemma \ref{lemma:NI} (see Lemma 7 in \cite{xiong2010negative}), there exists $ X= X^T > 0$ such that
\begin{equation*}
	\left[\begin{matrix}  X A+ A^T X &  X B- A^T C^T \\  B^T X- C A & -( C B+ B^T C^T)\end{matrix}\right] \leq 0.
\end{equation*}
This implies that for any $ z$, $x_1$, $x_2$ and $ v$, we have
\begin{equation}\label{eq:rd2 LMI z x v}
	\left[\begin{matrix}  z \\ x_1 \\ x_2 \\  v\end{matrix}\right]^T \left[\begin{matrix}  X A+ A^T X &  X B- A^T C^T \\  B^T X- C A & -( C B+ B^T C^T)\end{matrix}\right]\left[\begin{matrix}  z \\ x_1 \\ x_2 \\  v\end{matrix}\right] \leq 0.
\end{equation}
Let $\small  X=\left[\begin{matrix}  X_{11} &  X_{12} &  X_{13}\\  X_{12}^T &  X_{22} &  X_{23}\\  X_{13}^T &  X_{23}^T &  X_{33}\end{matrix}\right]$ and substitute the values of $ A$, $ B$ and $ C$ into (\ref{eq:rd2 LMI z x v}). Also, take  $x_1=0$, $x_2=0$ and $ v=- K_1 z$. We have
\begin{equation*}
	 z^T( X_{11} A_{11}+ A_{11}^T X_{11}) z \leq 0
\end{equation*}
for any $ z$, which implies that $ X_{11} A_{11}+ A_{11}^T X_{11}\leq 0$. Since $ X= X^T>0$, we have $ X_{11}>0$. Also, considering that $\det( A_{11})\neq 0$, according to Lemma \ref{thm:marginally stable}, $ A_{11}$ is Lyapunov stable. Also, the controllability of $( A, B)$ implies the controllability of $( A_{11}, A_{11} A_{13}+ A_{12})$. This completes the proof.
\end{proof}

\begin{remark}
The inverse of the matrix $ Y$ in (\ref{eq:Y rd2}) is
\begin{equation*}
 P= Y^{-1}=\left[\begin{matrix} {\mc P}_1 & {\mc P}_1  A_{11}^{-1} A_{12} & 0 \\  A_{12}^T A_{11}^{-T}{\mc P}_1 & {\mc P}_2 +  A_{12}^T A_{11}^{-T}{\mc P}_1  A_{11}^{-1} A_{12} & 0\\
0 & 0& I\end{matrix}\right],	
\end{equation*}
where ${\mc P}_1= {\mc Y}_1^{-1}>0$ and ${\mc P}_2={\mc Y}_2^{-1}>0$. This enables us to define a storage function for the system (\ref{eq:linearised4}) as
\begin{equation*}
V_2( z,x_1,x_2)=\frac{1}{2}	\left[\begin{matrix}  z^T & x_1^T & x_2^T\end{matrix}\right] P\left[\begin{matrix}  z \\ x_1 \\ x_2\end{matrix}\right],
\end{equation*}
which satisfies
\begin{equation*}
\dot V_2( z,x_1,x_2) \leq  v^T \dot y.
\end{equation*}
\end{remark}

\begin{remark}
In the special case of $ m=0$, the system (\ref{eq:linearised3}) is
\begin{subequations}\label{eq:rd2 no z}
\begin{align}
	\dot x_1 =&\ x_2,\\
	\dot x_2=&\  A_{32}x_1+  A_{33}x_2+\mc C \mc A \mc B u,\\
	y=&\ x_1.
\end{align}	
\end{subequations}
Let the input be
\begin{equation*}
u=(\mc C \mc A \mc B)^{-1}\left( v_0+(K_{01}-A_{32})x_1+ (K_{02}-A_{33}) x_2\right).
\end{equation*}
Then, the system (\ref{eq:rd2 no z}) becomes
\begin{subequations}\label{eq:rd2 no z final}
\begin{align}
	\dot x_1 =&\ x_2,\\
	\dot x_2 =&\ K_{01}x_1+K_{02}x_2+v_0,\\
	y=&\ x_1.
\end{align}
\end{subequations}
Choose $K_{01}$ to be such that $K_{01}=K_{01}^T<0$ and choose $K_{02}$ to be such that $K_{02}+K_{02}^T<0$. The matrix $\small Y_0=\left[\begin{matrix} -K_{01}^{-1} & 0 \\ 0 & I\end{matrix}\right]$ then satisfies the conditions in (\ref{eq:NI condition}), while the other conditions in Lemma \ref{lemma:NI} are all satisfied. Therefore, there exist matrices $K_{01}$ and $K_{02}$ such that he system (\ref{eq:rd2 no z final}) is an NI system.
\end{remark}

\subsection{Main Theorem}
To summarize the NI state feedback equivalence results for the relative degree one and two cases, we recall the following terminologies (see \cite{khalil2002nonlinear,isidori2013nonlinear}).

The systems (\ref{eq:linearised1}) and (\ref{eq:linearised3}) are said to be the \emph{normal forms} of the system (\ref{eq:original state-space}) in the relative degree one and the relative degree two cases, respectively. For these two cases, the dynamics described in (\ref{eq:linearised1a}) and (\ref{eq:linearised3a}) are not controlled by the input $u$ directly or through chains of integrators, and are called the \emph{internal dynamics}. Setting the other states to be zero in the internal dynamics, we obtain the \emph{zero dynamics}. That is $\dot z = A_{11}z$ for both relative degree one and two cases, with a minor abuse of notation. We now give the definitions of the weakly minimum phase property.

\begin{definition}(Weakly Minimum Phase)\cite{saberi1990global,byrnes1991passivity}
The system (\ref{eq:original state-space}) is said to be weakly minimum phase if its zero dynamics is Lyapunov stable.
\end{definition}

We now combine the NI state feedback equivalence result shown in Lemmas \ref{thm:rd1 feedback NI} and \ref{thm:rd2 feedback NI} in the following theorem.

\begin{theorem}\label{thm:main}
Suppose a system with the state-space model (\ref{eq:original state-space}) is of relative degree one or relative degree two and has no zero at the origin. Then it is state feedback equivalent to an NI system if and only if it is controllable and weakly minimum phase.
\end{theorem}
\begin{proof}
Via the use of state transformations, the controllability of the realization (\ref{eq:original state-space}) is retained in the realizations (\ref{eq:linearised1}) or (\ref{eq:linearised3}). As can be proved similarly to the proofs of Lemmas \ref{thm:rd1 feedback NI} and \ref{thm:rd2 feedback NI}, the controllabilities of the systems (\ref{eq:linearised1}) and (\ref{eq:linearised3}) are equivalent to those of the pairs $(A_{11},A_{12})$ and $( A_{11}, A_{11} A_{13}+ A_{12})$, respectively, for the relative degree one and two cases.
Because the system (\ref{eq:original state-space}) has no zero at the origin, then $\det(A_{11})\neq 0$ for both relative degree one and two cases. The rest of the proof follows directly from Lemmas \ref{thm:rd1 feedback NI} and \ref{thm:rd2 feedback NI}.
\end{proof}

\section{STATE FEEDBACK EQUIVALENCE TO A STRONGLY STRICT NEGATIVE IMAGINARY SYSTEM}\label{section:SSNI}

We now consider necessary and sufficient conditions under which the system (\ref{eq:original state-space}) is state feedback equivalent to an SSNI system. State feedback equivalence to an SSNI system is defined as follows:

\begin{definition}
The system (\ref{eq:original state-space}) is said to be state feedback equivalent to an SSNI system if there exists a state feedback control law
\begin{equation*}
u=K_x x+ K_v v	
\end{equation*}
such that the closed-loop system with the new input $v\in \mathbb R^p$ is an SSNI system.
\end{definition}

Similarly, we first consider the existence of state feedback matrices for systems of relative degree one and two in the normal forms.

\begin{lemma}\label{thm:rd1 SSNI}
Suppose the system (\ref{eq:linearised2}) has $(A_{11},A_{12})$ controllable. Then the following statements are equivalent:

1. $A_{11}$ is Hurwitz;

2. There exist $K_1$ and $K_2$ such that the system (\ref{eq:linearised2}) is an SSNI system with realisation $( A,B,C)$, where $  A$ is Hurwitz, and the transfer function $R(s):=C(sI-  A)^{-1}B$ is such that $R(s)+R(-s)^T$ has full normal rank.
\end{lemma}
\begin{proof}
Let the matrices $A$, $B$ and $C$ be the same as in (\ref{eq:rd1 A}), (\ref{eq:rd1 B}) and (\ref{eq:rd1 C}), respectively. Here, the matrices $K_1$ and $K_2$ in $A$ are different than those defined in Section \ref{section:rd1}.
From the proof of Lemma \ref{thm:rd1 feedback NI}, $( A,B)$ is controllable if and only if $(A_{11},A_{12})$ is controllable. Therefore, there is no observable uncontrollable mode in this system.

\textbf{Sufficiency}. According to Lemma \ref{thm:Lyapunov}, we can always find a matrix ${\mc Y}_1$ such that
\begin{equation*}
A_{11}{\mc Y}_1 + {\mc Y}_1 A_{11}^T + \frac{1}{2}A_{11}^{-1}A_{12}A_{12}^TA_{11}^{-T}<0	
\end{equation*}
is satisfied. In the sequel, we will find a matrix $ K_1$ such that
\begin{equation}\label{eq:rd1 SSNI Lyapunov ineq}
	A_{11}{\mc Y}_1 + {\mc Y}_1 A_{11}^T + \frac{1}{2}({\mc Y}_1 K_1^T+A_{11}^{-1}A_{12})( K_1{\mc Y}_1+A_{12}^TA_{11}^{-T})<0
\end{equation}
is satisfied. One possible choice is $ K_1=-A_{12}^TA_{11}^{-T}{\mc Y}_1^{-1}$, which simplifies (\ref{eq:rd1 SSNI Lyapunov ineq}) to be $A_{11}{\mc Y}_1 + {\mc Y}_1 A_{11}^T<0$. Let $ K_2 =  K_1A_{11}^{-1}A_{12}-{\mc Y}_2^{-1}$, where ${\mc Y}_2 \in \mathbb R^{p\times p}$ can be any symmetric positive definite matrix; i.e., ${\mc Y}_2={\mc Y}_2^T >0$. We apply Lemma \ref{thm:SSNI} in the following to prove that the system (\ref{eq:linearised2}) is an SSNI system. We construct the matrix $ Y$ as follows:
\begin{equation*}\label{eq:hat Y}
 Y=\left[\begin{matrix}{\mc Y}_1 + A_{11}^{-1}A_{12}{\mc Y}_2 A_{12}^TA_{11}^{-T} & -A_{11}^{-1}A_{12}{\mc Y}_2\\-{\mc Y}_2 A_{12}^TA_{11}^{-T} & {\mc Y}_2\end{matrix}\right].
\end{equation*}
We have $ Y>0$ because ${\mc Y}_2>0$ and the Schur complement of the block ${\mc Y}_2$ is $ {\mc Y}_1$, which is positive definite. Now, we have $B+A YC^T=0$ and
\begin{equation*}
 A Y+ Y A^T = \left[\begin{matrix}A_{11}{\mc Y}_1 + {\mc Y}_1 A_{11}^T & {\mc Y}_1  K_1^T+A_{11}^{-1}A_{12} \\  K_1{\mc Y}_1+A_{12}^TA_{11}^{-T} & -2I\end{matrix}\right].
\end{equation*}
We have $-2I<0$ and the Schur complement of the block $2I$ in the matrix $-( A Y+ Y A^T)$ is
\begin{align*}
	(-( A &  Y+ Y  A^T)) /(2I)\notag\\
	=& -A_{11}{\mc Y}_1 - {\mc Y}_1 A_{11}^T\notag\\
	 &- \frac{1}{2}({\mc Y}_1  K_1^T+A_{11}^{-1}A_{12})( K_1{\mc Y}_1+A_{12}^TA_{11}^{-T})>0,
\end{align*}
according to (\ref{eq:rd1 SSNI Lyapunov ineq}). Hence $ A Y+ Y A^T<0$. According to Lemma \ref{thm:Lyapunov}, $ A^T$ is Hurwitz. Therefore $ A$ is Hurwitz. We now prove that $R(s)+R(-s)^T$ has full normal rank. Given $ A$, $B$ and $C$, we have
\begin{align}
	R(s)=&\ C(sI- A)^{-1}B\notag\\
	=&\left[\begin{matrix} 0 & I\end{matrix}\right]\left[\begin{matrix} sI-A_{11} & -A_{12}\\ - K_1 & sI- K_2\end{matrix}\right]^{-1}\left[\begin{matrix} 0\\ I\end{matrix}\right]\notag\\
	=&\left(sI- K_1(sI-A_{11})^{-1}A_{12}- K_2\right)^{-1}.\label{eq:rank arbitrary s}
\end{align}
Take $s=0$ into (\ref{eq:rank arbitrary s}), we have
\begin{equation*}
R(0)=( K_1A_{11}^{-1}A_{12}- K_2)^{-1}={\mc Y}_2>0.
\end{equation*}
Hence $R(s)+R(-s)^T$ has full normal rank. Therefore, according to Lemma \ref{thm:SSNI}, the system (\ref{eq:linearised2}) is SSNI.

\textbf{Necessity}. If $ A$ is Hurwitz, $R(s)+R(-s)^T$ has full normal rank and the system (\ref{eq:linearised2}) is SSNI, then according to Lemma \ref{thm:SSNI}, there exists a matrix $ Y= Y^T>0$ such that $B=- A YC^T$ and $ A Y+ Y A^T<0$.

Let $ X= Y^{-1}$, then $ X= X^T>0$. Let $Q = -( A Y+ Y A^T)$, then we have $Q=Q^T>0$ and $ X A+ A^T X = - X Q  X <0$. Since $B=- A YC^T$, we have $CB+B^TC^T=-C A YC^T-C Y A^TC^T=C QC^T$. Also, $ XB- A^TC^T=- X A X^{-1}C^T- A^TC^T=-( X A+ A^T X) X^{-1}C^T= XQ X X^{-1}C^T= XQC^T$. Since $Q=Q^T>0$, let $H:=Q^{\frac{1}{2}}$. Hence $H=H^T>0$. We have
\begin{align}
	\left[\begin{matrix}  X A+ A^T X &  XB- A^TC^T \\ B^T X-C A & -(CB+B^TC^T)\end{matrix}\right]=& -\left[\begin{matrix} L^T \\ W^T\end{matrix}\right]\left[\begin{matrix} L & W\end{matrix}\right]\notag\\
	\leq &\ 0,\label{eq:rd1 LMI SSNI}
\end{align}
where $L=H  X$ and $W= -H C^T$. (\ref{eq:rd1 LMI SSNI}) implies that for any $z\in \mathbb R^{m}$, $y\in \mathbb R^{p}$ and $v\in \mathbb R^{p}$, we have
\begin{align}
	&\left[\begin{matrix} z^T & y^T & v^T\end{matrix}\right]\left[\begin{matrix}  X A+ A^T X &  XB- A^TC^T \\ B^T X-C A & -(CB+B^TC^T)\end{matrix}\right]\left[\begin{matrix} z \\ y \\ v\end{matrix}\right]\notag\\
	&=\left[\begin{matrix} z^T & y^T & v^T\end{matrix}\right]\left[\begin{matrix} L^T \\ W^T\end{matrix}\right]\left[\begin{matrix} L & W\end{matrix}\right]\left[\begin{matrix} z \\ y \\ v\end{matrix}\right] \leq 0,\label{eq:rd1 LMI z y v SSNI}
\end{align}
where $``="$ holds if and only if $\small\left[\begin{matrix} L & W\end{matrix}\right]\left[\begin{matrix} z \\ y \\ v\end{matrix}\right]=0$. That is $L\left[\begin{matrix} z \\ y\end{matrix}\right]+Wv=0$, which is $\small H\left( X\left[\begin{matrix} z \\ y\end{matrix}\right]-C^Tv\right)=0$. Because $H>0$, this equation holds if and only if
\begin{equation}\label{eq:rd1 z y v eq 0}
	 X\left[\begin{matrix} z \\ y\end{matrix}\right]-C^Tv=0.
\end{equation}
Let $\small  X=\left[\begin{matrix}  X_{11} &  X_{12}\\  X_{12}^T &  X_{22}\end{matrix}\right]$ and choose $y=0$ and $v=- K_1z$. Then (\ref{eq:rd1 z y v eq 0}) becomes
\begin{equation*}
	\left[\begin{matrix}  X_{11} \\  X_{12}^T+ K_1\end{matrix}\right]z = 0,
\end{equation*}
which holds only if $ X_{11}z=0$. Since $ X= X^T>0$, $ X_{11}= X_{11}^T>0$. Hence $ X_{11}z=0\iff z=0$. This implies that with the choice $y=0$ and $v=- K_1z$, strict inequality holds for (\ref{eq:rd1 LMI z y v SSNI}) for all $z\neq 0$. Take (\ref{eq:rd1 A}), (\ref{eq:rd1 B}) and (\ref{eq:rd1 C}) together with $y=0$ and $v=- K_1z$ into (\ref{eq:rd1 LMI z y v SSNI}), we get
\begin{equation*}
z^T( X_{11}A_{11}+A_{11}^T X_{11})z<0	
\end{equation*}
for all $z\neq 0$. This implies that $ X_{11}A_{11}+A_{11}^T X_{11}<0$. Therefore, according to Lemma \ref{thm:Lyapunov}, $A_{11}$ is Hurwitz.
\end{proof}

\begin{remark}\label{remark:rd2 SSNI}
Unlike the relative degree one case, the system (\ref{eq:linearised4}) can not be an SSNI system. Indeed, considering the particular form of $B$ and $C$ as are specified in (\ref{eq:rd2 B}) and (\ref{eq:rd2 C}), the condition $B+AYC^T=0$ in (\ref{eq:SSNI condition}) requires the middle diagonal block of $AY$ be $0$. Therefore, the matrix $AY+YA^T$ can never be sign definite. Hence, the system (\ref{eq:linearised4}) can never be SSNI.
\end{remark}

Therefore, we conclude the SSNI state feedback equivalence result in the following theorem. First we give the definition of the minimum phase property.

\begin{definition}(Minimum Phase)\cite{byrnes1991passivity,khalil2002nonlinear}
The system (\ref{eq:original state-space}) is said to be minimum phase if its zero dynamics is asymptotically stable.
\end{definition}

\begin{theorem}\label{thm:main SSNI}
Consider a system with the state-space model (\ref{eq:original state-space}), suppose it is controllable and has relative degree one. Then following statements are equivalent:

1. The system is minimum phase;

2. The system is state feedback equivalent to an SSNI system with realisation $( A,B,C)$, where $  A$ is Hurwitz, and the transfer function $R(s):=C(sI-  A)^{-1}B$ is such that $R(s)+R(-s)^T$ has full normal rank.
\end{theorem}
\begin{proof}
	The controllability of the system (\ref{eq:original state-space}) implies the controllability of $(A_{11},A_{12})$ in the system (\ref{eq:linearised2}). The system (\ref{eq:original state-space}) is minimum phase if and only if $A_{11}$ in the system (\ref{eq:linearised1}) is Hurwitz. The rest of the proof follows from Lemma \ref{thm:rd1 SSNI}.
\end{proof}

\section{CONTROL OF SYSTEMS WITH SNI UNCERTAINTY}\label{section:synthesis}
\begin{figure}[h!]
\centering
\psfrag{delta}{$\Delta(s)$}
\psfrag{nominal}{\hspace{-0.05cm}Nominal}
\psfrag{plant}{Plant}
\psfrag{controller}{Controller}
\psfrag{closed-loop}{Closed-Loop}
\psfrag{w}{$w$}
\psfrag{x}{$x$}
\psfrag{y}{$y$}
\psfrag{u}{$u$}
\psfrag{R_s}{$R(s)$}
\psfrag{+}{$+$}
\includegraphics[width=8.5cm]{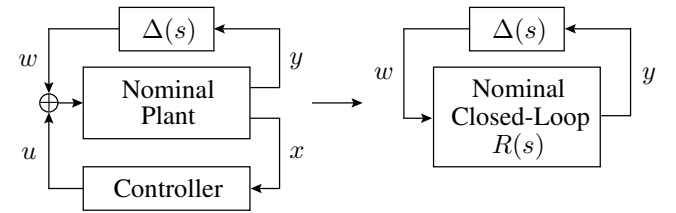}
\caption{A feedback control system. The plant uncertainty $\Delta(s)$ is SNI and satisfies $\lambda_{max}(\Delta(0))\leq \gamma$ and $\Delta(\infty)\geq 0$. Under some assumptions, we can find a controller such that the closed-loop transfer function $R(s)$ is NI with $R(\infty)=0$ and $\lambda_{max}(R(0))< 1/\gamma$. Then the closed-loop system is robust stable.}
\label{fig:controller synthesis}
\end{figure}
Consider the uncertain feedback control system in Fig.~\ref{fig:controller synthesis}. Suppose full state feedback is available. Then Lemmas \ref{thm:rd1 feedback NI} and \ref{thm:rd2 feedback NI} can be used to synthesize a state-feedback controller such that the closed-loop system is NI. In fact, a similar problem is addressed in \cite{petersen2010feedback}, where the existence of such state-feedback controllers relies on the solvability of a series of LMIs. However, in this paper, the LMI assumptions in \cite{petersen2010feedback} are replaced by some simpler assumptions with respect to the controllability and Lyapunov stability of some matrices in the state-space model of the system.

Let us consider the system
\begin{subequations}\label{eq:original uncertain system}
\begin{align}
	\dot x =&\ \mc Ax+\mc B(u+w),\label{eq:uncertain A}\\
	y =&\ \mc C x,\label{eq:uncertain B}\\
	w =&\ \Delta(s)y,
\end{align}	
\end{subequations}
where the uncertainty transfer function $\Delta(s)$ is assumed to be SNI with $\Delta(\infty)\geq 0$ and $\lambda_{max}(\Delta(0))\leq \gamma$ for some constant $\gamma<\infty$. We consider the relative degree one and two cases separately.

\subsection{Relative Degree One Case}
Suppose the system (\ref{eq:original uncertain system}) has relative degree one, that is $\det(\mc C \mc B)\neq 0$. Then without loss of generality, the system (\ref{eq:original uncertain system}) can be considered to be in the SCB (see \cite{sannuti1987special,chen2004linear})
\begin{subequations}\label{eq:rd1 uncertain}
\begin{align}
\dot z =&\ A_{11}z+A_{12}y,\\
\dot y =&\ A_{21}z+A_{22}y+\mc C \mc B(u+w),\\
y=&\ [\begin{matrix}0& {I}\end{matrix}]\left[\begin{matrix}z \\ y\end{matrix}\right],\\
w =&\ \Delta(s)y.
\end{align}
\end{subequations}
With the result in Section \ref{section:rd1}, the following stabilization theorem is obtained.
\begin{lemma}\label{lemma:uncertain rd1}
Suppose the uncertain system (\ref{eq:rd1 uncertain})	 satisfies $\det(A_{11})\neq 0$, $A_{11}$ Lyapunov stable and $(A_{11},A_{12})$ controllable. Then the system (\ref{eq:rd1 uncertain}) can be stabilized by the state-feedback control law
\begin{equation}\label{eq:rd1 stabilizing input}
u=(\mc C\mc B)^{-1}((K_1-A_{21})z+(K_2-A_{22})y),
\end{equation}
where $K_1$ is defined in (\ref{eq:rd1 K_1}) and $K_2$ is defined in (\ref{eq:K_2 definition}) with $\mc Y_2$ also satisfying $\lambda_{max}(\mc Y_2)<\frac{1}{\gamma}$.
\end{lemma}
\begin{proof}
With the input (\ref{eq:rd1 stabilizing input}) applied, the system (\ref{eq:rd1 uncertain}) becomes a positive feedback interconnection of the system $w=\Delta(s)y$ and $y=R(s)w$, where $R(s)$ is the transfer function of the state-space model (\ref{eq:linearised2}) with $v$ replaced by $w$. Therefore, according to Lemma \ref{thm:rd1 feedback NI}, the state-feedback matrices $K_1$ and $K_2$ as defined respectively in (\ref{eq:rd1 K_1}) and (\ref{eq:K_2 definition}) make $R(s)$ negative imaginary. Also, we have $R(\infty)=0$ and $R(0)=\mc Y_2$. Because $\lambda_{max}(\mc Y_2)<\frac{1}{\gamma}$ and $\lambda_{max}(\Delta(0))\leq \gamma$, then $\lambda_{max}(R(0)\Delta(0))<1$.

Thus, we have $R(\infty)\Delta(\infty)=0$, $\Delta(\infty)\geq 0$ and $\lambda_{max}(R(0)\Delta(0))<1$. According to Lemma \ref{lemma:dc gain theorem}, the positive feedback interconnection $[R(s),\Delta(s)]$ is internally stable. Therefore, the closed-loop system (\ref{eq:rd1 uncertain}) with input (\ref{eq:rd1 stabilizing input}) is robustly stable.
\end{proof}

\begin{remark}
For the uncertain system (\ref{eq:rd1 uncertain}), if the uncertainty $\Delta(s)$ is NI, then it can be robustly stabilized by applying Lemma \ref{thm:rd1 SSNI} in a similar way to make the nominal closed-loop system SSNI.
\end{remark}

\subsection{Relative Degree Two Case}
Suppose the system (\ref{eq:original uncertain system}) has relative degree two, that is $\mc C \mc B=0$ and $\det(\mc C \mc A \mc B)\neq 0$. Then without loss of generality, the system (\ref{eq:original uncertain system}) can be considered to be in the SCB (see \cite{sannuti1987special,chen2004linear})
\begin{subequations}\label{eq:rd2 uncertain}
\begin{align}
\dot { z} =&\  A_{11} z+ A_{12}x_1+ A_{13}x_2,\\
\dot x_1 =&\ x_2,\\
\dot x_2 =&\  A_{31} z +  A_{32} x_1+ A_{33}x_2+\mc C\mc A\mc B (u+w)\\
y =& \left[\begin{matrix}0& I & 0\end{matrix}\right]\left[\begin{matrix} z \\ x_1 \\ x_2\end{matrix}\right],\\
w=&\ \Delta(s)y.
\end{align}\end{subequations}
With the result in Section \ref{section:rd2}, the following stabilization theorem is obtained.
\begin{lemma}\label{thm:rd2 synthesis}
Suppose the uncertain system (\ref{eq:rd2 uncertain})	 satisfies $\det(A_{11}) \neq 0$, $A_{11}$ Lyapunov stable and $(A_{11},A_{11}A_{13}+A_{12})$ controllable. Then the system (\ref{eq:rd2 uncertain}) can be stabilized by the state-feedback control law
\begin{equation}\label{eq:rd2 stabilizing input}
u=(\mc C\mc A \mc B)^{-1}((K_1-A_{31})z+(K_2-A_{32})x_1+(K_3-A_{33})x_2),
\end{equation}
where $K_1$ is defined in (\ref{eq:rd2 K_1}), $K_3$ is defined in (\ref{eq:rd2 K_3}) and $K_2$ is defined in (\ref{eq:rd2 K_2 definition}) with $\mc Y_2$ also satisfying $\lambda_{max}(\mc Y_2)<\frac{1}{\gamma}$.
\end{lemma}
\begin{proof}
With the input (\ref{eq:rd2 stabilizing input}) applied, the system (\ref{eq:rd2 uncertain}) becomes a positive feedback interconnection of the system $w=\Delta(s)y$ and $y=R(s)w$, where $R(s)$ is the transfer function of the state-space model (\ref{eq:linearised4}) with $v$ replaced by $w$. Therefore, according to Lemma \ref{thm:rd2 feedback NI}, the state-feedback matrices $K_1$, $K_2$ and $K_3$ as defined respectively in (\ref{eq:rd2 K_1}), (\ref{eq:rd2 K_2 definition}) and (\ref{eq:rd2 K_3}) make $R(s)$ negative imaginary. Also, we have $R(\infty)=0$ and $R(0)=\mc Y_2$. Because $\lambda_{max}(\mc Y_2)<\frac{1}{\gamma}$ and $\lambda_{max}(\Delta(0))\leq \gamma$, then $\lambda_{max}(R(0)\Delta(0))<1$.

Thus, we have $R(\infty)\Delta(\infty)=0$, $\Delta(\infty)\geq 0$ and $\lambda_{max}(R(0)\Delta(0))<1$. According to Lemma \ref{lemma:dc gain theorem}, the positive feedback interconnection $[R(s),\Delta(s)]$ is internally stable. Therefore, the closed-loop system (\ref{eq:rd2 uncertain}) with input (\ref{eq:rd2 stabilizing input}) is robustly stable.
\end{proof}

\subsection{Existence of a Stabilizing State Feedback Control Law}
The results in Lemmas \ref{lemma:uncertain rd1} and \ref{thm:rd2 synthesis} are concluded in the following theorem.

\begin{theorem}
Consider the uncertain system (\ref{eq:original uncertain system}), suppose it has relative degree one or two, $(\mc A,\mc B)$ is controllable and the realization $(\mc A,\mc B,\mc C)$ has no zero at the origin. Also, suppose the realization $(\mc A,\mc B,\mc C)$ is weakly minimum phase. Then there always exists state feedback in the form $u=Kx$ that asymptotically stabilizes this system. The formulas for such $K$ is provided in (\ref{eq:rd1 stabilizing input}) and (\ref{eq:rd2 stabilizing input}) for the relative one and two cases, respectively.
\end{theorem}
\begin{proof}
The proof follows directly from	Lemmas \ref{lemma:uncertain rd1} and \ref{thm:rd2 synthesis}.
\end{proof}

\section{ILLUSTRATIVE EXAMPLE}\label{section:example}
Consider an uncertain system with the state-space model
\begin{subequations}\label{eq:example}
\begin{align}
\dot x =& \left[\begin{matrix}
-1&1&0\\1&-1&1\\0&1&-1	
\end{matrix}
\right]x	+ \left[\begin{matrix}
0\\0\\1	
\end{matrix}
\right] w + \left[\begin{matrix}
0\\0\\1	
\end{matrix}
\right] u,\\
y=& \left[\begin{matrix}
0&1&0	
\end{matrix}
\right]x,\\
w=&\ \Delta(s)y,
\end{align}
\end{subequations}
where the transfer function $\Delta(s)$ is SNI with $\lambda_{max}(\Delta(0))<1$ and $\Delta(\infty)\geq 0$. Let
\begin{equation*}
	\mc A=\left[\begin{matrix}
-1&1&0\\1&-1&1\\0&1&-1	
\end{matrix}
\right], \quad
\mc B=\left[\begin{matrix}
0\\0\\1	
\end{matrix}
\right], \quad
\mc C=\left[\begin{matrix}
0&1&0	
\end{matrix}
\right].
\end{equation*}
The nominal plant with the state-space realisation $(\mc A,\mc B,\mc C)$ is not an NI system because $\mc A$ is unstable. Therefore, we need to apply the proposed state feedback equivalence result to make it NI. We have $\mc C \mc B=0$ and $\mc C \mc A \mc B=1$. Hence, the system (\ref{eq:example}) has relative degree two. With a state transformation
\begin{equation*}
	\left[\begin{matrix}
z\\ x_1\\x_2	
\end{matrix}\right] = Tx, \quad \textnormal{where} \quad T=\left[\begin{matrix}
1 & 0 & 0\\ 0 & 1 & 0\\1 & -1 & 1	
\end{matrix}\right],
\end{equation*}
the system (\ref{eq:example}) becomes
\begin{subequations}\label{eq:example SCB}
\begin{align}
	\dot z=&-z+x_1,\\
	\dot x_1=&\ x_2,\\
	\dot x_2=&\ x_1-2x_2+w+u,\\
	y =&\ x_1,\\
	w=&\ \Delta(s)y.
\end{align}	
\end{subequations}
In comparison to the system (\ref{eq:rd2 uncertain}), here we have $A_{11}=-1$, $A_{12}=1$, $A_{13}=0$, $A_{31}=0$, $A_{32}=1$, $A_{33}=-2$. The assumptions in Lemma \ref{thm:rd2 synthesis} that $A_{11}$ is Lyapunov stable and $(A_{11},A_{11}A_{13}+A_{12})$ is controllable are satisfied. According to (\ref{eq:rd2 K_1}), (\ref{eq:rd2 K_2 definition}) and (\ref{eq:rd2 K_3}), choose the state-feedback matrices to be
\begin{equation*}
K_1= 1,\ K_2= -1-\mc Y_2^{-1}, \ \textnormal{and}\  K_3=-1.
\end{equation*}
Choose $\mc Y_2=0.5<\frac{1}{\lambda_{max}(\Delta(0))}$, then $K_2=-3$. According to (\ref{eq:rd2 stabilizing input}), let
\begin{equation}\label{eq:example control law}
u=z_1-4x_1+x_2.
\end{equation}
Then the system (\ref{eq:example SCB}) becomes
\begin{subequations}\label{eq:example final}
\begin{align}
	\dot z=&\ -z+x_1,\label{eq:example SCB z}\\
	\dot x_1=&\ x_2,\\
	\dot x_2=&\ z_1-3x_1-x_2+w,\\
	y =&\ x_1,\label{eq:example SCB y}\\
	w=&\ \Delta(s)y.
\end{align}	
\end{subequations}
The transfer function of the nominal closed-loop system described by (\ref{eq:example SCB z})-(\ref{eq:example SCB y}) is
\begin{equation*}
R(s)=\frac{s+1}{s^3+2s^2+4s+2},	
\end{equation*}
which has a Bode plot shown in Fig.~\ref{fig:bode}. Since $\angle R(s)\in [-\pi,0]$ for positive frequencies, $R(s)$ is NI. Also, the magnitude of the DC gain of $R(s)$ is less than unity. In fact, $R(0)=\frac{1}{2}$. Therefore, $\lambda_{max}(R(0)\Delta(0))<1$. Because $R(\infty)\Delta(\infty)=0$ and $\Delta(\infty)\geq 0$, the system (\ref{eq:example final}) is asymptotically stable. Thus, the system (\ref{eq:example SCB}) is robustly stabilized by the control law (\ref{eq:example control law}).

\begin{figure}[h!]
\centering
\psfrag{Bode Diagram}{\hspace{-0.5cm}\small Bode Plot of $R(s)$}
\psfrag{Magnitude (dB)}{\hspace{-0.2cm}\footnotesize{Magnitude (dB)}}
\psfrag{Phase (deg)}{\hspace{0.1cm}\footnotesize{Phase (\textdegree)}}
\psfrag{Frequency  (rad/s)}{\hspace{-0.1cm}\footnotesize Frequency (rad/s)}
\psfrag{0.01}{\tiny$10^{-2}$}
\psfrag{0.11}{\tiny$10^{-1}$}
\psfrag{0.12}{\tiny$10^{0}$}
\psfrag{0.13}{\tiny$10^1$}
\psfrag{0.14}{\tiny$10^{2}$}
\psfrag{0}{\tiny$0$}
\psfrag{-60}{\hspace{-0.05cm}\tiny$-60$}
\psfrag{-20}{\hspace{-0.05cm}\tiny$-20$}
\psfrag{-80}{\hspace{-0.05cm}\tiny$-80$}
\psfrag{-40}{\hspace{-0.05cm}\tiny$-40$}
\psfrag{-45}{\hspace{-0.05cm}\tiny$-45$}
\psfrag{-90}{\hspace{-0.05cm}\tiny$-90$}
\psfrag{-135}{\hspace{-0.05cm}\tiny$-135$}
\psfrag{-180}{\hspace{-0.05cm}\tiny$-180$}
\hspace{-0.6cm}
\includegraphics[width=9.2cm]{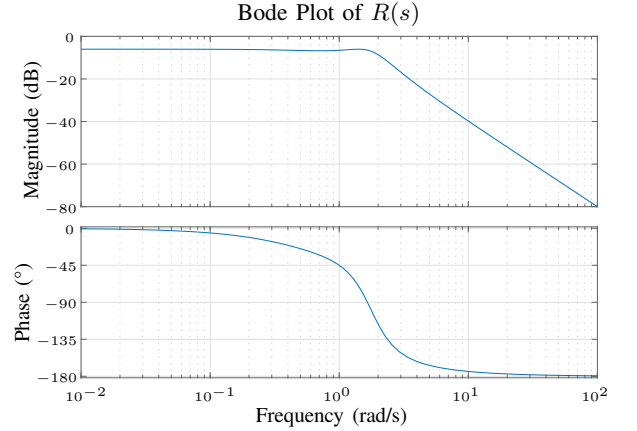}
\caption{Bode plot of the closed-loop transfer function $R(s)$ from the uncertainty input $w$ to the uncertainty output $y$. This closed-loop system is obtained from the system (\ref{eq:example SCB}) using a full state-feedback control law (\ref{eq:example control law}) obtained from Lemma \ref{thm:rd2 synthesis}.}
\label{fig:bode}
\end{figure}

%% file: Conclusion.tex
\section{CONCLUSION AND FUTURE WORK}
\label{section:conclusion}
In this paper, we have provided necessary and sufficient conditions under which a system of relative degree one or two is state feedback equivalent to an NI system. As is stated in Theorem \ref{thm:main}, the system (\ref{eq:original state-space}), which is of relative degree one or two and has no zeros at the origin, is state feedback equivalent NI if and only if it is controllable and weakly minimum phase. A similar SSNI feedback equivalence result is presented in Theorem \ref{thm:main SSNI}. The state feedback NI results are then applied to solve the robust stabilization problem for an uncertain system with a specific uncertainty. An example is also provided to illustrate stabilizing process for an uncertain system.

The results of this paper have been recently extended by the authors in \cite{shi2021necessary}. The paper \cite{shi2021necessary} completes the present paper by considering the case when a system has mixed relative degree one and two, while the present paper considers the relative degree one and relative degree two cases separately. It is provided in \cite{shi2021necessary} the necessary and sufficient conditions for a system in the form of (\ref{eq:original state-space}) to be state feedback equivalent to an NI system.

Considering the emergence of the nonlinear negative imaginary systems theory (see \cite{ghallab2018extending,shi2021identical,shi2021free}), it is also worth investigating the feedback equivalence problem for nonlinear systems using the nonlinear NI systems theory. This future state feedback equivalent nonlinear NI research is planned to be a complement of the work done by Byrnes, Isidori and Williems in \cite{byrnes1991passivity}, which investigates the feedback passivity problem for a nonlinear system of relative degree one. It can be also regarded as an extension of the present paper to nonlinear systems.

%% file: main.bbl

%% file: main.bbl
\begin{thebibliography}{10}
\providecommand{\url}[1]{#1}
\csname url@samestyle\endcsname
\providecommand{\newblock}{\relax}
\providecommand{\bibinfo}[2]{#2}
\providecommand{\BIBentrySTDinterwordspacing}{\spaceskip=0pt\relax}
\providecommand{\BIBentryALTinterwordstretchfactor}{4}
\providecommand{\BIBentryALTinterwordspacing}{\spaceskip=\fontdimen2\font plus
\BIBentryALTinterwordstretchfactor\fontdimen3\font minus
  \fontdimen4\font\relax}
\providecommand{\BIBforeignlanguage}[2]{{%
\expandafter\ifx\csname l@#1\endcsname\relax
\typeout{** WARNING: IEEEtran.bst: No hyphenation pattern has been}%
\typeout{** loaded for the language `#1'. Using the pattern for}%
\typeout{** the default language instead.}%
\else
\language=\csname l@#1\endcsname
\fi
#2}}
\providecommand{\BIBdecl}{\relax}
\BIBdecl

\bibitem{lanzon2008stability}
A.~Lanzon and I.~R. Petersen, ``Stability robustness of a feedback
  interconnection of systems with negative imaginary frequency response,''
  \emph{IEEE Transactions on Automatic Control}, vol.~53, no.~4, pp.
  1042--1046, 2008.

\bibitem{petersen2010feedback}
I.~R. Petersen and A.~Lanzon, ``Feedback control of negative-imaginary
  systems,'' \emph{IEEE Control Systems Magazine}, vol.~30, no.~5, pp. 54--72,
  2010.

\bibitem{xiong2010negative}
J.~Xiong, I.~R. Petersen, and A.~Lanzon, ``A negative imaginary lemma and the
  stability of interconnections of linear negative imaginary systems,''
  \emph{IEEE Transactions on Automatic Control}, vol.~55, no.~10, pp.
  2342--2347, 2010.

\bibitem{preumont2018vibration}
A.~Preumont, \emph{Vibration control of active structures: an
  introduction}.\hskip 1em plus 0.5em minus 0.4em\relax Springer, 2018, vol.
  246.

\bibitem{halim2001spatial}
D.~Halim and S.~R. Moheimani, ``Spatial resonant control of flexible
  structures-application to a piezoelectric laminate beam,'' \emph{IEEE
  transactions on control systems technology}, vol.~9, no.~1, pp. 37--53, 2001.

\bibitem{pota2002resonant}
H.~Pota, S.~R. Moheimani, and M.~Smith, ``Resonant controllers for smart
  structures,'' \emph{Smart Materials and Structures}, vol.~11, no.~1, p.~1,
  2002.

\bibitem{cai2010stability}
C.~Cai and G.~Hagen, ``Stability analysis for a string of coupled stable
  subsystems with negative imaginary frequency response,'' \emph{IEEE
  Transactions on Automatic Control}, vol.~55, no.~8, pp. 1958--1963, 2010.

\bibitem{rahman2015design}
M.~A. Rahman, A.~Al~Mamun, K.~Yao, and S.~K. Das, ``Design and implementation
  of feedback resonance compensator in hard disk drive servo system: A mixed
  passivity, negative-imaginary and small-gain approach in discrete time,''
  \emph{Journal of Control, Automation and Electrical Systems}, vol.~26, no.~4,
  pp. 390--402, 2015.

\bibitem{bhikkaji2011negative}
B.~Bhikkaji, S.~R. Moheimani, and I.~R. Petersen, ``A negative imaginary
  approach to modeling and control of a collocated structure,'' \emph{IEEE/ASME
  Transactions on Mechatronics}, vol.~17, no.~4, pp. 717--727, 2011.

\bibitem{mabrok2013spectral}
M.~A. Mabrok, A.~G. Kallapur, I.~R. Petersen, and A.~Lanzon, ``Spectral
  conditions for negative imaginary systems with applications to
  nanopositioning,'' \emph{IEEE/ASME Transactions on Mechatronics}, vol.~19,
  no.~3, pp. 895--903, 2013.

\bibitem{das2014mimo}
S.~K. Das, H.~R. Pota, and I.~R. Petersen, ``A {MIMO} double resonant
  controller design for nanopositioners,'' \emph{IEEE Transactions on
  Nanotechnology}, vol.~14, no.~2, pp. 224--237, 2014.

\bibitem{das2014resonant}
------, ``Resonant controller design for a piezoelectric tube scanner: A mixed
  negative-imaginary and small-gain approach,'' \emph{IEEE Transactions on
  Control Systems Technology}, vol.~22, no.~5, pp. 1899--1906, 2014.

\bibitem{das2015multivariable}
------, ``Multivariable negative-imaginary controller design for damping and
  cross coupling reduction of nanopositioners: a reference model matching
  approach,'' \emph{IEEE/ASME Transactions on Mechatronics}, vol.~20, no.~6,
  pp. 3123--3134, 2015.

\bibitem{brogliato2007dissipative}
B.~Brogliato, R.~Lozano, B.~Maschke, and O.~Egeland, ``Dissipative systems
  analysis and control,'' \emph{Theory and Applications}, vol.~2, 2007.

\bibitem{song2012negative}
Z.~Song, A.~Lanzon, S.~Patra, and I.~R. Petersen, ``A negative-imaginary lemma
  without minimality assumptions and robust state-feedback synthesis for
  uncertain negative-imaginary systems,'' \emph{Systems \& Control Letters},
  vol.~61, no.~12, pp. 1269--1276, 2012.

\bibitem{mabrok2020dissipativity}
M.~A. Mabrok, M.~A. Alyami, and E.~E. Mahmoud, ``On the dissipativity property
  of negative imaginary systems,'' \emph{Alexandria Engineering Journal}, 2020.

\bibitem{kokotovic1989positive}
P.~Kokotovic and H.~Sussmann, ``A positive real condition for global
  stabilization of nonlinear systems,'' \emph{Systems \& Control Letters},
  vol.~13, no.~2, pp. 125--133, 1989.

\bibitem{saberi1990global}
A.~Saberi, P.~Kokotovic, and H.~Sussmann, ``Global stabilization of partially
  linear composite systems,'' \emph{SIAM Journal on Control and Optimization},
  vol.~28, no.~6, pp. 1491--1503, 1990.

\bibitem{byrnes1991passivity}
C.~I. Byrnes, A.~Isidori, J.~C. Willems \emph{et~al.}, ``Passivity, feedback
  equivalence, and the global stabilization of minimum phase nonlinear
  systems,'' \emph{IEEE Transactions on automatic control}, vol.~36, no.~11,
  pp. 1228--1240, 1991.

\bibitem{byrnes1991asymptotic}
C.~I. Byrnes and A.~Isidori, ``Asymptotic stabilization of minimum phase
  nonlinear systems,'' \emph{IEEE Transactions on Automatic Control}, vol.~36,
  no.~10, pp. 1122--1137, 1991.

\bibitem{santosuosso1997passivity}
G.~Santosuosso, ``Passivity of nonlinear systems with input-output
  feedthrough,'' \emph{Automatica}, vol.~33, no.~4, pp. 693--697, 1997.

\bibitem{lin1995feedback}
W.~Lin, ``Feedback stabilization of general nonlinear control systems: a
  passive system approach,'' \emph{Systems \& Control Letters}, vol.~25, no.~1,
  pp. 41--52, 1995.

\bibitem{jiang1996passification}
Z.-P. Jiang, D.~J. Hill, and A.~L. Fradkov, ``A passification approach to
  adaptive nonlinear stabilization,'' \emph{Systems \& Control Letters},
  vol.~28, no.~2, pp. 73--84, 1996.

\bibitem{khalil2002nonlinear}
H.~K. Khalil and J.~W. Grizzle, \emph{Nonlinear systems}.\hskip 1em plus 0.5em
  minus 0.4em\relax Prentice hall Upper Saddle River, NJ, 2002, vol.~3.

\bibitem{isidori2013nonlinear}
A.~Isidori, \emph{Nonlinear control systems}.\hskip 1em plus 0.5em minus
  0.4em\relax Springer Science \& Business Media, 2013.

\bibitem{petersen2016negative}
I.~R. Petersen, ``Negative imaginary systems theory and applications,''
  \emph{Annual Reviews in Control}, vol.~42, pp. 309--318, 2016.

\bibitem{lanzon2011strongly}
A.~Lanzon, S.~Patra, I.~R. Petersen, and Z.~Song, ``A strongly strict
  negative-imaginary lemma for non-minimal linear systems,''
  \emph{Communications in Information and Systems}, vol.~11, no.~2, pp.
  139--142, 2011.

\bibitem{bernstein2009matrix}
D.~S. Bernstein, \emph{Matrix mathematics: theory, facts, and formulas}.\hskip
  1em plus 0.5em minus 0.4em\relax Princeton university press, 2009.

\bibitem{hespanha2018linear}
J.~P. Hespanha, \emph{Linear systems theory}.\hskip 1em plus 0.5em minus
  0.4em\relax Princeton university press, 2018.

\bibitem{sannuti1987special}
P.~Sannuti and A.~Saberi, ``Special coordinate basis for multivariable linear
  systems-finite and infinite zero structure, squaring down and decoupling,''
  \emph{International Journal of Control}, vol.~45, no.~5, pp. 1655--1704,
  1987.

\bibitem{chen2004linear}
B.~M. Chen, Z.~Lin, and Y.~Shamash, \emph{Linear systems theory: a structural
  decomposition approach}.\hskip 1em plus 0.5em minus 0.4em\relax Springer
  Science \& Business Media, 2004.

\bibitem{shi2021necessary}
K.~Shi, I.~R. Petersen, and I.~G. Vladimirov, ``Necessary and sufficient
  conditions for state feedback equivalence to negative imaginary systems,''
  \emph{Submitted to IEEE Transactions on Automatic Control, available as arXiv
  preprint arXiv:2109.11273}, 2021.

\bibitem{ghallab2018extending}
A.~G. Ghallab, M.~A. Mabrok, and I.~R. Petersen, ``Extending negative imaginary
  systems theory to nonlinear systems,'' in \emph{2018 IEEE Conference on
  Decision and Control (CDC)}.\hskip 1em plus 0.5em minus 0.4em\relax IEEE,
  2018, pp. 2348--2353.

\bibitem{shi2021identical}
K.~Shi, I.~G. Vladimirov, and I.~R. Petersen, ``Robust output feedback
  consensus for networked identical nonlinear negative-imaginary systems,''
  \emph{IFAC-PapersOnLine}, vol.~54, no.~9, pp. 239--244, 2021.

\bibitem{shi2021free}
K.~Shi, I.~R. Petersen, and I.~G. Vladimirov, ``Output feedback consensus for
  networked heterogeneous nonlinear negative-imaginary systems with free body
  motion,'' \emph{Submitted to IEEE Transactions on Automatic Control,
  available as arXiv preprint arXiv:2011.14610}, 2021.

\end{thebibliography}
